\documentclass[12pt,notitlepage]{article}
\usepackage{amsthm}
\usepackage[utf8]{inputenc}
\usepackage[english]{babel}
\usepackage[T1]{fontenc}
\usepackage{hyperref}
\usepackage{amsmath}
\usepackage{amsfonts}
\usepackage{xcolor} 
\usepackage{caption}
\usepackage{verbatim}
\usepackage{array,multirow,makecell}
\setcellgapes{1pt}
\makegapedcells
\newcolumntype{R}[1]{>{\raggedleft\arraybackslash }b{#1}}
\newcolumntype{L}[1]{>{\raggedright\arraybackslash }b{#1}}
\newcolumntype{C}[1]{>{\centering\arraybackslash }b{#1}}

\newcommand{\tr}{\mathop{\text{tr}}\nolimits}
\newcommand{\val}{\mathop{\text{val}}\nolimits}

\newtheorem{proposition}{Proposition}
\newtheorem{conjecture}{Conjecture}
\newtheorem{corollary}{Corollary}

\newtheorem{lemma}{Lemma}
\usepackage{enumerate}
\usepackage{subfig}
\usepackage{appendix}

\newcommand{\cA}{{\mathcal A}}
\newcommand{\cB}{{\mathcal B}}

\newcommand{\cF}{{\mathcal F}}
\newcommand{\cG}{{\mathcal G}}
\newcommand{\cH}{{\mathcal H}}

\newcommand{\cS}{{\mathcal S}}

\def\uN{\text{U}(N)}

\def\oN{\text{O}(N)}
\def\oD{\text{O}(D)}

\def\cvp{\raise 2pt\hbox{,}}

\newcommand{\bP}{{\mathbf P}}

\definecolor{mygray}{gray}{0.3}

\newcommand\beq{\begin{equation}}
\newcommand\eeq{\end{equation}}
\newcommand{\bes}{\begin{eqnarray}}
\newcommand{\ees}{\end{eqnarray}}

\def\vphi{{\varphi}}

\newcommand\restr[2]{{
  \left.\kern-\nulldelimiterspace 
  #1 
  \vphantom{\big|} 
  \right|_{#2} 
  }}

\usepackage{multicol}
\usepackage{amssymb}
\usepackage{graphicx}
\usepackage[left=2cm,right=2cm,top=2.5cm,bottom=2.5cm]{geometry}

\begin{document}

\
\vfill

\begin{center}
\textbf{\Large{On the large $D$ expansion of}}

\smallskip

\textbf{\Large{Hermitian multi-matrix models}}
\vspace{15pt}

\vfill

{\large Sylvain Carrozza,$^{a,}$\footnote{\url{scarrozza@perimeterinstitute.ca}} Frank Ferrari,$^{b,}$\footnote{\url{frank.ferrari@ulb.ac.be}} Adrian Tanasa$^{c,}$\footnote{\url{ntanasa@u-bordeaux.fr}}}

{\large and Guillaume Valette$^{b,}$\footnote{\url{guillaume.valette@ulb.ac.be}}}

\vspace{10pt}

$^{a}${\sl Perimeter Institute for Theoretical Physics\\
 31 Caroline St N, Waterloo, ON N2L 2Y5, Canada\\
}

\vspace{3pt}

$^{b}${\sl Service de Physique Th\'eorique et Math\'ematique\\
Universit\'e Libre de Bruxelles (ULB) and International Solvay Institutes\\
Campus de la Plaine, CP 231, B-1050 Brussels, Belgium, EU}

\vspace{3pt}

$^{c}${\sl LaBRI, Univ.\ Bordeaux\\
351 cours de la Lib\'{e}ration, 33405 Talence, France, EU\\
H.\  Hulubei  Nat.\  Inst.\  Phys.\  Nucl.\  Engineering,  \\
P.O.Box  MG-6,  077125  Magurele,  Romania,  EU\\ I.\  U.\  F.,  1  rue  Descartes,  75005  Paris,  France, EU
}


\end{center}

\vspace{5pt}

\begin{abstract}

\noindent We investigate the existence and properties of a double asymptotic expansion in $1/N^{2}$ and $1/\sqrt{D}$ in $\uN\times\oD$ invariant Hermitian multi-matrix models, where the $N\times N$ matrices transform in the vector representation of $\oD$. The crucial point is to prove the existence of an upper bound $\eta(h)$ on the maximum power $D^{1+\eta(h)}$ of $D$ that can appear for the contribution at a given order $N^{2-2h}$ in the large $N$ expansion. We conjecture that $\eta(h)=h$ in a large class of models. In the case of traceless Hermitian matrices with the quartic tetrahedral interaction, we are able to prove that $\eta(h)\leq 2h$; the sharper bound $\eta(h)= h$ is proven for a complex bipartite version of the model, with no need to impose a tracelessness condition. We also prove that $\eta(h)=h$ for the Hermitian model with the sextic wheel interaction, again with no need to impose a tracelessness condition.

\end{abstract}

\vfill\eject

\setcounter{footnote}{0}

\section{Introduction}

Since its introduction by 't~Hooft in 1974 \cite{tHooftplanar}, the large $N$ expansion of matrix models has been extensively studied and has found a wealth of applications in theoretical and mathematical physics, from QCD to random geometry, string theory and black hole physics; see e.g.\ \cite{revMM} for useful books and reviews. The case of multi-matrix models is particularly important but also very difficult. Except in some very special instances \cite{revsolveMM}, multi-matrix models remain intractable, even in the leading, planar, large $N$ approximation.

We focus on an interesting class of so-called Hermitian matrix-vector models. These models have a $\uN\times\oD$ symmetry, with Hermitian matrices $X_{\mu}$, $1\leq\mu\leq D$, transforming in the adjoint representation of $\uN$ and in the vector representation of $\oD$.\footnote{Other cases, for instance $\oN\times\oD$ invariant models with real symmetric or antisymmetric matrices, can be discussed similarly.} The extra $\oD$ symmetry allows one to study the large $D$ limit on top of the usual large $N$ limit, to further simplify the problem. 

The most straightforward implementation of this idea, which is to take the limit $D\rightarrow\infty$ at fixed $N$, yields a rather trivial result: the Feynman graphs at leading order are ``trees of bubbles,'' subleading orders are given by a loopwise expansion with loop expansion parameter $1/D$, and the matrix model effectively degenerates to a vector model \cite{vectorrev}. A much more interesting large $D$ limit was introduced in \cite{Frank}. The idea is to enhance the large $D$ scaling of some of the 't~Hooft's couplings, in a way that will be reviewed below, so that more graphs contribute at each order in the expansion. The goal is to find a new approximation scheme which is more tractable than the original large $N$ limit but which is also able to capture the correct qualitative non-perturbative physics of large $N$ matrix models. In several interesting cases, the leading graphs are melons or generalized melons \cite{Frank,FRV}. This property was discovered before in tensor models (see e.g.\ \cite{revtensors1,Bonzom_2011,revtensors2,revtensors3,CT,Witten,revtensors4}) and some disordered models in condensed matter physics (see e.g.\ \cite{revCM}) and has been at the origin of many recent developments (see e.g.\ \cite{revdev} for reviews). Note that the relation between our matrix models and tensor models is not surprising: our matrices have a three-index structure $X^{a}_{\mu\, b}$ and, even if the indices $a,b$ and $\mu$ transform under distinct symmetry groups, $X$ may formally be viewed as a tensor of rank three.

The enhancement of the couplings at large $D$ introduces a crucial new feature: the large $N$ and the large $D$ limits do not commute. In fact, the $D\rightarrow\infty$ limit is ill-defined at fixed $N$, because one can find contributions to physical observables that are proportional to an arbitrarily high power of $D$. Yet, the large $D$ limit can still be defined consistently, in the following way. 

Focusing on the free energy for concreteness, we first take $N\rightarrow\infty$ and get the usual $1/N^2$ expansion
\beq\label{genusExp} \mathcal{F}(\lambda;N,D) = \sum_{h\in\mathbb{N}} N^{2-2h} F_{h}(\lambda;D)\, .\eeq
In the above formula, $\lambda$ denotes collectively all the coupling constants of the model, which are independent of $N$ and $D$ after taking into account the appropriate 't~Hooft's large $N$ and enhanced large $D$ scalings. The contributions $F_{h}$ can themselves be expanded perturbatively in a power series of the couplings,
\beq\label{perturbexp} F_{h}(\lambda;D)=\sum_{k\in\mathbb N}F_{h}^{(k)}(D)\lambda^{k}\, ,\eeq
where $\lambda^{k}$ denotes any degree $k$ monomial in the couplings. Each coefficient $F_{h}^{(k)}(D)$ is given by a sum over Feynman graphs with $k$ vertices. Note that this sum is always finite, because there is a finite number of graphs with a fixed number of vertices. At large $D$, $F_{h}^{(k)}(D)$ scales as $D^{1+\eta_{k}(h)}$ and we define
\beq\label{chidef}\boxed{\eta(h)=\sup_{k\in\mathbb N}\eta_{k}(h)\, .}\eeq
Due to the form of the enhanced large $D$ scaling of the coupling constants, we will see in the models studied in this paper that the $\eta_{k}(h)$ must be integers or half-integers. If $\eta(h)$ is finite, each term $F_{h}$ then has a well-defined large $D$ expansion, with expansion parameter $1/\sqrt{D}$, of the form
\beq\label{largeDexp} F_{h}(\lambda;D)=\sum_{\ell\in\mathbb N}D^{1+\eta(h)-\frac{\ell}{2}}F_{h,\ell}(\lambda)\, .\eeq
The main purpose of our work is to evaluate $\eta(h)$. We conjecture that
\beq\label{conjecture} \boxed{\eta(h) = h}\eeq
in a large class of interesting models and we check this conjecture in some interesting non-trivial examples. Note that the function $\eta$ is not bounded above, which explains why the large $D$ limit cannot be defined at fixed $N$. Note also that the existence of the large $D$ expansion is ensured as soon as one can prove that $\eta(h)<\infty$. In the particularly interesting case of the traceless Hermitian matrix model with the tetrahedral interaction term $\tr X_{\mu}X_{\nu}X_{\mu}X_{\nu}$, we shall prove that $\eta(h)\leq 2h$, which is enough to ensure the consistency of the large $D$ limit at any order in the $1/N^{2}$ expansion. We believe that the sharper result \eqref{conjecture} is valid in this case, but this remains unproven. The sharp bound $\eta(h)=h$ will be proven for a complex bipartite version of the tetrahedral model as well as for the Hermitian model with the wheel interaction term $\tr X_{\mu}X_{\nu}X_{\rho}X_{\mu}X_{\nu}X_{\rho}$, with no need to impose the $\tr X_{\mu}=0$ condition.

The conjecture \eqref{conjecture} has been proven in full generality when the $\uN\times\oD$ symmetry of the models is enhanced to $\uN_{\text L}\times\uN_{\text R}\times\oD$, with the $\uN_{\text L}\times\uN_{\text R}$ acting as $X_{\mu}\mapsto U_{\text L}X_{\mu}U_{\text R}^{-1}$ \cite{Frank}. This is of course possible only if the Hermitian constraint is waived and the matrices are complex. In this case, the indices $a$, $b$ and $\mu$ in $X^{a}_{\mu\, b}$ transform independently and are associated with distinguishable cycles (or faces) in the stranded representation of Feynman graphs. Equivalently, the Feynman graphs have a so-called colored graph representation which can be analysed using techniques developed and reviewed e.g.\ in \cite{revtensors1,revtensors2,revtensors3,Frank,Azeyanagi:2017mre,FRV}. In particular, in \cite{FRV}, a very general theory of tensor and matrix-tensor models is presented, which works for any rank and any interaction term, as long as each index of the tensors or the matrix-tensors transform with respect to distinct symmetry groups.

The relation \eqref{conjecture} also allows the definition of an interesting double scaling limit $N\rightarrow\infty$, $D\rightarrow\infty$, $D/N^{2}=\text{constant}$, that we briefly sketch in Section \ref{sec:dble} and that has already been investigated in detail in the $\uN_{\text L}\times\uN_{\text R}\times\oD$ context \cite{double-scaling-largeD}. 

Our goal in the present paper is to focus on the Hermitian case, for which only the diagonal $\uN$ subgroup of $\uN_{\text L}\times\uN_{\text R}$ is preserved and thus for which the colored graph technology cannot be used.

Going beyond the realm of colored graphs is a big and highly non-trivial leap forward, for both tensor and matrix-vector models. Deprived of the beautiful colored graph technology, the existence of the large $N$ (for tensors) or large $D$ (for matrix-vectors) limit looks like a miracle. The first evidence that the results obtained with colored graphs could be generalized to the case of Hermitian matrices was provided in \cite{Frank}, where it was shown that the large $D$ limit of the Hermitian models, including in particular models with the tetrahedral interaction $\tr X_{\mu}X_{\nu}X_{\mu}X_{\nu}$, exists at the planar level, with $\eta(0)=0$, and that the leading Feynman graphs still have the standard melonic structure. However, the argument of \cite{Frank} cannot be generalized beyond the planar case and it was actually known that $\eta(1)=\infty$ due to the chain-of-tadpoles graphs depicted in Fig.\ \ref{fig:tadpole_chain}. The same class of graphs also ruin the large $N$ limit of a model of symmetric tensors with the tetrahedral interaction, for instance. In an interesting paper \cite{Klebanov:2017nlk}, Klebanov and Tarnopolsky conjectured that this problem could be cured by working with traceless symmetric tensors. This was proven in \cite{Benedetti:2017qxl}, following a pioneering work by Gurau \cite{Gurau:2017qya} on a particular bipartite model for which the tracelessness condition is irrelevant. The results of \cite{Benedetti:2017qxl} actually apply to any rank-three tensor model with tetrahedral interaction, for which the tensor transforms in an irreducible representation of $\oN$ or $\mathrm{Sp}(N)$ (including mixed representations of the permutation group) \cite{Carrozza:MixedPerm,Carrozza:2018psc}.

The proofs in \cite{Gurau:2017qya} and \cite{Benedetti:2017qxl} rely on a detailed, often tedious, combinatorial analysis of the Feynman graphs.  
We do not know yet of a general method that allows to determine whether a particular model supports a well-defined large $N$ expansion; results available in the literature still rely on case by case investigations.
Our analysis below of the large $D$ limit of matrix-vector models is no different. Even though one can easily state general conjectures like \eqref{conjecture} in a large class of models, see Section \ref{sec:complex}, our present understanding requires to invent proofs on a case-by-case basis. The case of the traceless Hermitian model with the tetrahedral interaction $\tr X_{\mu}X_{\nu}X_{\mu}X_{\nu}$ is discussed in Section \ref{sec:Hermitian}. We are able to derive that $\eta(h)\leq 2h$. Some of the results in \cite{Benedetti:2017qxl} are used in a crucial way to obtain this upper bound. This proves the existence of the large $D$ limit in general, extending the result of \cite{Frank} from planar graphs to any order in the $1/N^{2}$ expansion. Proving the sharp bound \ref{conjecture} remains a difficult open question in this case. We obtain it for a version of the model for which the matrices $X_{\mu}$ are complex with no tracelessness condition and for which the Feynman graphs have a bipartite structure, with interaction term $\lambda\tr X_{\mu}X_{\nu}X_{\mu}X_{\nu}+\text{H.\ c.}$. In section \ref{sec:hermitiansix}, we prove the sharp bound \ref{conjecture} for the Hermitian model with the wheel interaction $\tr X_{\mu}X_{\nu}X_{\rho}X_{\mu}X_{\nu}X_{\rho}$, with no need to impose the tracelessness condition. To the best of our knowledge, this is the first time an interaction of degree greater than four is treated beyond the scope of colored graph techniques. We also briefly discuss some much simpler models at the end of Section \ref{sec:gensecherm}, normal-ordered or bipartite with no maximally single trace interactions, for which the existence of the large $D$ limit can be easily derived.


\section{\label{sec:complex}Models and conjectures}

We limit ourselves to models with a single matrix-vector $X_{\mu}$ and with single-trace interaction terms. The discussion is essentially unchanged if several matrices $X_{\mu}$, $Y_{\mu}$, etc., are included: our analysis of the Feynman graphs can be repeated by forgetting about the distinction between the different matrix-vectors. The number of space-time dimensions is also irrelevant, and thus we work in zero dimension for simplicity.

\subsection{\label{reviewFRV}Brief review of the $\uN^{2}\times\oD$ symmetric models}

The models with $\uN_{\text L}\times\uN_{\text R}\times\oD$ symmetry \cite{Frank} have an action of the form
\beq\label{ModelColored}
S = ND \, \Bigl( \tr X_\mu^\dagger X_\mu  +\sum_{a} \lambda_a D^{g(\cB_a)} I_{\cB_a}(X,X^\dagger) \Bigr) \, ,\eeq
where the interaction terms read
\beq\label{IBdef1} I_{\cB_a} = \tr X_{\mu_{1}}X_{\mu_{2}}^{\dagger}X_{\mu_{3}}X_{\mu_{4}}^{\dagger}X_{\mu_{5}}\cdots X_{\mu_{2s_{a}}}^{\dagger}\, ,\eeq
with pairwise identifications and summations over the $\oD$ indices $\mu_{i}$. These identifications are encoded in a 3-colored graph, also called a $3$-bubble, $\cB_a$. Note that $2s_{a}$ is the order of the interaction, associated with the coordination number, or valency, of the interaction vertices in Feynman graphs. The coupling is enhanced by the factor $D^{g(\cB_{a})}$ compared to the standard large $D$ scaling. This enhancement is governed by the so-called genus $g(\cB_{a})$ of the 3-bubble $\cB_{a}$. We refer to \cite{Frank,FRV} for detailed definitions. For our purposes, all we need to know is that the genus takes the form
\beq\label{GenusColored2}
g(\cB_a)=\frac{s_a-1}{2} - \frac{x(\cB_a)}{2} \, \cvp\eeq
where $x(\cB_{a})$ is a non-negative integer. The interaction terms for which $x(\cB_{a})=0$ form a privileged class called maximally single-trace (MST) \cite{FRV}. The MST interactions are characterized by the fact that they can be written as single-traces, when the traces are computed with respect to any choice of pair of indices, including the mixed choices for which we pair $\uN$ and $\oD$ indices. The paradigmatic example is the so-called tetrahedral interaction
\beq\label{tetradef} I_{\cB_{T}}(X,X^{\dagger})=\tr X_{\mu}X_{\nu}^{\dagger}X_{\mu}X_{\nu}^{\dagger}\, .\eeq
If we introduce the matrices $Y^{a}$ and $Z_{b}$, with matrix elements $(Y^{a})_{b\mu} =(Z_{b})^{a}_{\ \mu}=X^{a}_{\mu\, b}$, together with their complex congugates and transposes,
it is straightforward to check the MST property: $I_{\cB_{T}} = \tr Y^{a}\, {}^{\text T} Y^{b}\,\bar Y_{a}\, Y_{b}^{\dagger}=\tr Z_{a}\, {}^{\text T}Z_{b}\,\bar Z^{a}\, Z^{\dagger}_{b}$. It is easy to check that the tetrahedron is the only MST interaction of order four, and that there are two possibilities at order six, the so-called wheel and prism interactions,
\begin{align}\label{wheeldef} & I_{\cB_{W}}(X,X^{\dagger})=\tr X_{\mu}X_{\nu}^{\dagger}X_{\rho}X_{\mu}^{\dagger}X_{\nu}X_{\rho}^{\dagger}\, ,\\
\label{prismdef} & I_{\cB_{P}}(X,X^{\dagger})=\tr X_{\nu}X_{\mu}^{\dagger}X_{\nu}X_{\rho}^{\dagger}X_{\mu}X_{\rho}^{\dagger}\, .\end{align}

The existence of the large $D$ limit of the models \eqref{ModelColored} was proven in \cite{Frank}. For all these models, the large $N$ counting parameter $h$ coincides with the genus of the graphs and it is known \cite{FRV} that the relation \eqref{conjecture} is true if one includes a MST interaction term. This can be generalized to $\oN\times\oD^{r}$ invariant models with multi-trace interactions \cite{FRV}. These results also imply similar results for standard tensor models after setting $D=N$, including the multi-orientable tensor model of \cite{revtensors3}. The expansions \eqref{genusExp} and \eqref{largeDexp}, with $\eta(h)=h$, then combine into a single expansion
\beq \label{Model1ExpLargeN}
\mathcal F(\lambda;N)  =  \sum_{\omega\in\frac{1}{2}\mathbb{N}} N^{3-\omega} F_{\omega}(\lambda) \, ,\eeq
where the half-integer $\omega = h + \frac{\ell}{2}$ coincides with the index defined in \cite{FRV} or equivalently with the degree defined in \cite{CT} in the quartic model. Note finally that the leading graphs contributing at order $N^{2}D$ are the generalized melons defined in \cite{Frank,FRV}. Their complete classification, which has been achieved only in special cases \cite{CT,FRV}, remains an open probem.

\subsection{\label{sec:conj}Conjectures}

A natural and ambitious generalization of the above results is to consider models of traceless Hermitian matrices $X_{\mu}$ with the general action \eqref{ModelColored} in which we set $X_{\mu}^{\dagger}=X_{\mu}$. A strong conjecture for these models is that the large $D$ limit exists at all orders in the $1/N^{2}$ expansion, with Eq.\ \eqref{conjecture}, $\eta(h)=h$, being valid as soon as we include a MST interaction term. A weaker version of the conjecture is simply that $\eta(h)<\infty$ for all $h$, which ensures the consistency of the large $D$ expansion at all orders in $1/N^{2}$. We can also consider more restrictive versions of these conjectures, for instance by limiting ourselves to models in which only one MST interaction term is turned on. We prove below that $\eta(h)\leq 2h$ for the model with the tetrahedral interaction \eqref{tetradef} and that $\eta(h)=h$ for the model with the wheel interaction \eqref{wheeldef}. In the latter case, there is no need to impose that $\tr X_{\mu}=0$.

Another technically interesting generalization is to consider models of complex matrices $X_{\mu}$, with no traceless constraint, but for which the Feynman graphs are bipartite. This restriction allows to consider only a subclass of the graphs relevant in the Hermitian case, for which the analysis is significantly simpler. These models have an action
\beq\label{ModelBipartiteColored}
S = ND \, \Bigl( \tr X_\mu^\dagger X_\mu +\sum_{a} D^{g(\cB_a)}\bigl( \lambda_a I_{\cB_a}(X,X)+ \text{H.\ c.}\bigr) \Bigr) \, .\eeq
Note that even though the matrices are complex, the interaction terms respect only the diagonal subgroup of $\uN_{\text L}\times\uN_{\text R}$. In Section \ref{sec:bitetra}, we show the strong form of the conjecture, $\eta(h)=h$, for the tetrahedral model of this type. At the end of Section \ref{sec:gensecherm}, we also show by a simple argument that the large $D$ limit exists at all orders in the $1/N^{2}$ expansion for any bipartite model that does not include MST interaction terms.

\subsection{\label{sec:dble}Double scaling limit}

An interesting consequence of the conjecture \eqref{conjecture}, or more generally of the existence of a linear law $\eta(h)=\kappa h$, $\kappa>0$, is to allow the definition of a double scaling limit
\beq\label{NDdblescale} N\rightarrow\infty\, ,\quad D\rightarrow\infty\, ,\quad \frac{N}{D^{\kappa/2}}= M < \infty\, ,\eeq
for which the expansions \eqref{genusExp} and \eqref{largeDexp} combine into
\beq \label{eq:free-en-0}
\underset{ \substack{N,D\to\infty \\ M<\infty} }{\mathrm{lim}} \frac{1}{D^{1+\kappa}} \cF(\lambda;N,D) = \sum_{h\geq 0} M^{2-2h}  F_{h,0}(\lambda)\, .
\eeq
The leading order in this expansion is dominated by the $\ell=0$ graphs \emph{of any genus} and thus encodes informations about the matrix model at all orders in $1/N^{2}$. A discussion of this expansion and the classification of the $\ell=0$ graphs in the $\uN_{\text L}\times\uN_{\text{R}}$ symmetric tetrahedral model has recently appeared \cite{double-scaling-largeD}.


\section{Traceless Hermitian tetrahedral interaction}
\label{sec:Hermitian}

We consider the traceless Hermitian model with the tetrahedral interaction term \eqref{tetradef}. After performing a convenient rescaling of the matrix variables, the action reads 
\beq \label{Model2ActionRescaled}
S = \frac{1}{2}\tr X_\mu X_\mu + \frac{\lambda}{4 N\sqrt{D}} \, \tr X_\mu X_\nu X_\mu X_\nu \, .\eeq
Taking into account the tracelessness condition, the free propagator is given by
\beq \label{Model2FreeProp}
\bigl\langle X^{a}_{\mu\, b} X^{c}_{\nu\, d} \bigr\rangle_0 = \delta_{\mu\nu}\Bigl( \delta^a_d \delta^c_b - \frac{1}{N} \delta^a_b \delta^c_d \Bigr) := \bP^{a}_{\mu b, }{}^{c}_{\nu d}
\, .\eeq
We focus on the free energy (vacuum amplitude) $\mathcal F(\lambda;N,D)$, which has the usual formal matrix integral representation
\beq \label{modelPartitionFunction}
e^{-\mathcal F(\lambda;N,D)} = \int [\text{d} X] \, e^{-S} \,,
\eeq
where $[\text{d} X]$ is the standard $\uN$-invariant measure on Hermitian matrices. Note that this is not restrictive: the existence of the large $D$ expansion for correlation functions can be easily derived from the analysis of the connected vacuum graphs.

\subsection{\label{sec:ganda}Graphs and amplitudes}

\subsubsection{Feynman graphs, stranded graphs, primary and descendants}

\begin{figure}
\centerline{\includegraphics[scale=1]{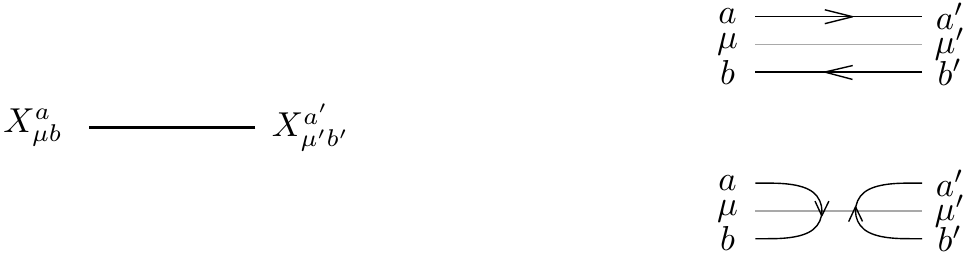}}
\vskip 1cm
\centerline{\includegraphics[scale=1]{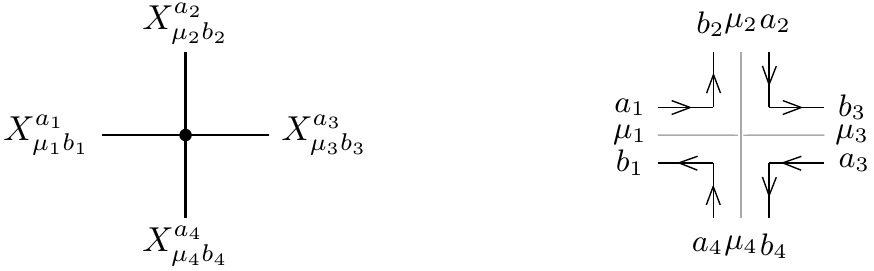}}
\caption{Edge and vertex for the model \eqref{Model2ActionRescaled}, in the non-stranded (left inset) and stranded (right inset) representations. In the stranded representation, the edge split into the sum of two contributions (unbroken and broken) associated with the two terms in the propagator \eqref{Model2FreeProp}.\label{fig:EdgeVertex}}
\end{figure}

Connected Feynman graphs, which are denoted as $\cG$, are built from edges and quadrivalent vertices, see Fig.\ \ref{fig:EdgeVertex}, left inset. The total number of edges $e(\cG)$ and vertices $v(\cG)$ satisfy $e(\cG)=2v(\cG)$. Due to the matrix structure, the cyclic ordering of the edges attached to a given vertex is important. In other words, the Feynman graphs are embedded on orientable surfaces and we associate to them a genus $g(\cG)\in\mathbb{N}$. The edges and the vertices also have a convenient stranded graph representation, in which the strands associated with the $\uN$ and $\oD$ indices are explicitly displayed, see Fig.\ \ref{fig:EdgeVertex}, right inset. The $\uN$ strands form ribbons and are oriented, because the fundamental representation of $\uN$ is complex. The edges of the Feynman graphs are given by a sum of two stranded contributions, unbroken and broken, associated with the two terms in \eqref{Model2FreeProp}. This implies that to a given Feynman graph $\cG$ is associated a set $\mathcal S(\cG)$ of $2^{e(\cG)}$ stranded graphs. We call \emph{primary graph} of $\cG$, noted $G_{*}$, its unique stranded configuration with only unbroken edges, and \emph{descendant graph} any other stranded configuration of $\cG$ obtained from $G_{*}$ by cutting ribbon edges; see Fig.\ \ref{fig:example_descendents} for an example. Finally, to any stranded graph $G$ is associated the ribbon graph $\hat G$, obtained from $G$ by removing the $\oD$ strands.

\begin{figure}
\centerline{\includegraphics[scale=.5]{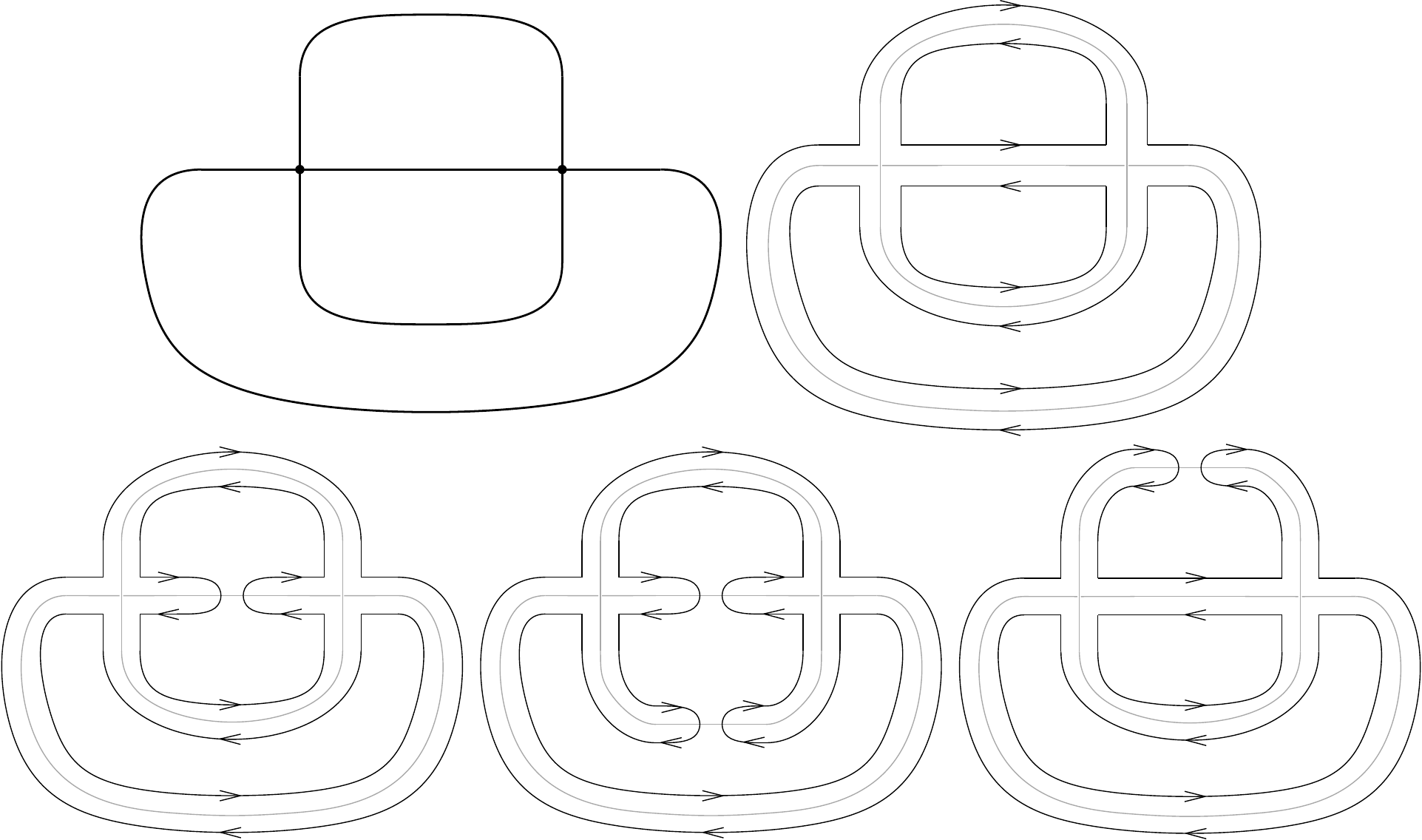}}
    \caption{Top: a Feynman graph and its primary stranded configuration; bottom: three of the $2^4 - 1 = 15$ descendants.\label{fig:example_descendents}}
\end{figure}

We can make the following useful remarks and definitions.
\begin{itemize}
\item The non-stranded representation of the Feynman graphs contains all the relevant information to construct all the associated stranded graphs: the $\uN$ ribbon structure is completely fixed by the cyclic ordering around of the edges attached to each vertex and the $\oD$ strands must go straight-ahead through the vertices.

\item The number of vertices of a stranded graph $G$ associated to $\cG$ is $v(G)=v(\cG)$. $G$ has $b(G)$ broken edges and $u(G)$ unbroken edges, with $b(G)+u(G)=e(G)=e(\cG)$. The ribbon graph $\hat G$ has $v_{4}(\hat G)=v(\cG)$ tetravalent vertices, $v_{1}(\hat G)=2b(G)$ univalent vertices (inserted at the end of cut ribbons), $e(\hat G)=u(G)+2b(G)$ edges and $c(\hat G)$ connected components. Note the relation $2 e(\hat G) = v_{1}(\hat G) + 4 v_{4}(\hat G)$.

\item We call \emph{faces} the closed cycles of $\uN$ strands and simply \emph{cycles} the closed cycles of $\oD$ strands. Each stranded graph $G$ has a certain number $f(G)$ of faces and $\varphi(G)$ of cycles. Its genus $g(G) := g(\hat G)$ is given by Euler's formula, 
\beq\label{EulerGhat} 2c(\hat G) - 2 g(G) = f(G)- e(\hat{G}) +( v_1(\hat{G})+ v_4(\hat{G})) = f(G)-u(G)+v(G)\, .\eeq
The genus of the Feynman graph is the genus of its primary stranded graph: $g(\cG)=g(G_{*})$. Note that cutting ribbon edges can only decrease the genus, so $g(G)\leq g(\cG)$. Note also that $\varphi(G)$ is the same for all stranded graphs and thus we can define $\varphi(\cG) := \varphi(G)$ for any $G\in\mathcal S(\cG)$.

\item A \emph{Feynman $2$-point graph} is a vacuum Feynman graph in which an edge has been cut into two external half-edges. An \emph{elementary tadpole $2$-point graph} is a Feynman $2$-point graph with only one vertex; and an \emph{elementary melon $2$-point graph} is a Feynman $2$-point graph with exactly two vertices connected to each other by three edges. The elementary tadpole and melon 2-point graphs are depicted in Fig.\ \ref{fig:tadpolemelon}.
\end{itemize}
\begin{figure}
\centerline{\includegraphics[scale=1]{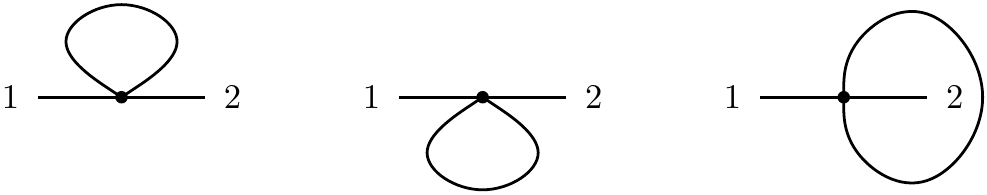}}
\vskip 1cm
\centerline{\includegraphics[scale=1]{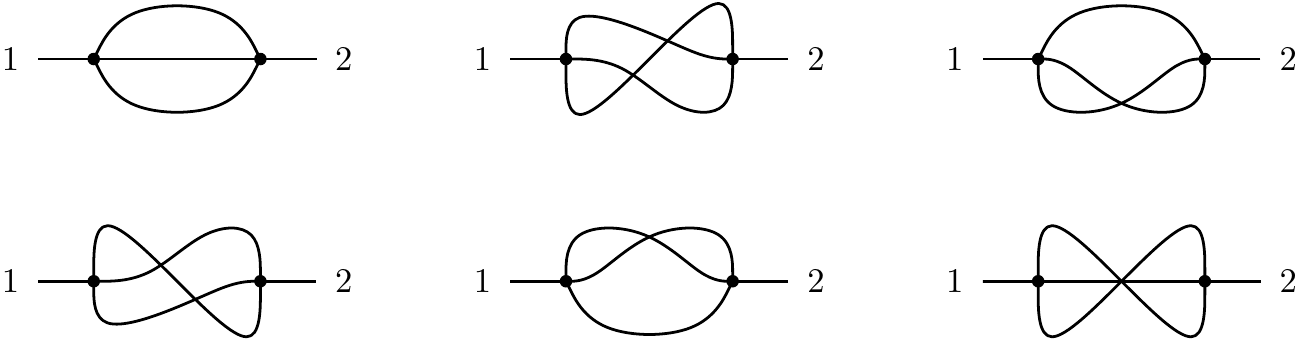}}
\caption{Elementary tadpole and melon 2-point graphs.\label{fig:tadpolemelon}}
\end{figure}
%


\subsubsection{Amplitudes}

The amplitude associated to a Feynman graph $\cG$ is given by a sum over the $2^{e(\cG)}$ stranded configurations of $\cG$,\footnote{Overall constant factors, such as symmetry factors associated to $\cG$, are irrelevant for our discussion; we thus discard them. On the other hand, the relative signs and constant factors of the various stranded configurations $G$ of $\cG$ are crucial.}
\beq \label{Model2AmplitudeFeynman}
\mathcal{A}(\cG) = \sum_{G \in \cS(\cG)} \mathcal{A}(G)\, .\eeq
The $N$- and $D$-dependence of individual stranded graphs is given by a simple monomial term and reads
\beq \label{Model2AmplitudeStranded}
\mathcal{A}(G) =  
\lambda^{v(\cG)} (-1)^{b(G)} N^{2-2h(G)}D^{1+\alpha(\cG)}\, ,
\eeq
where we have introduced the parameters
\beq \label{Model2hExp1}
h(G) := 1+\frac{1}{2}\bigl(v(G)-f(G)+ b(G)\bigr) \, ,
\eeq
and
\beq \label{Model2alphaExp1}
\alpha(\cG) := -\frac{1}{2}v(\cG)+\varphi(\cG) - 1 \,.
\eeq
Using $u+b=2v$ and \eqref{EulerGhat}, $h$ may be rewritten as
\beq \label{Model2hExp2}
h (G) =g(G)+1 - c(\hat G) + b(G) \, .\eeq
This shows that $h$ is an integer and, since $c(\hat G)\leq b(G)+1$ (we create at most one new connected component each time we cut a ribbon edge), $h(G)\geq g(G)$. For the primary graph, $h$ coincides with the genus, but this is not true in general. The rule is that, if one cuts a \emph{regular} ribbon edge (a ribbon edge for which the two $\uN$ strands belong to two distinct faces), one does not change the genus nor the number of connected components and thus increases $h$ by one unit; and if one cuts a \emph{singular} ribbon edge (a ribbon edge for which the two $\uN$ strands belong to the same face), one either decreases the genus by one unit or increases $c$ by one unit, thus keeping $h$ fixed. Overall,
\beq\label{hgineq} h(G)\geq h(G_{*})=g(\cG)\, .\eeq

Noting that the powers of $D$ and $\lambda$ in \eqref{Model2AmplitudeStranded} are the same for all stranded graphs associated with $\cG$, the above discussion shows that we can write $\mathcal A(\cG) = \lambda^{v(\cG)}N^{2}D^{1+\alpha(\cG)}\tilde P_{\cG}(N^{-2})$ for a certain polynomial $\tilde P_{\cG}$. Since the tracelessness condition implies that all amplitudes must vanish when $N=1$ (in particular, the propagator \eqref{Model2FreeProp} identically vanish when $N=1$), $\tilde P_{\cG}(1)=0$ and it is convenient to factor out explicitly $N^{2}-1$:
\beq \label{modelAmplitudeFeynman2}
\mathcal{A}(\cG) = \lambda^{v(\cG)} (N^2 -1) D^{1+\alpha(\cG)} P_{\cG}(N^{-2}) \, .
\eeq
From \eqref{hgineq}, we know that the valuation of the polynomial $P_{\cG}$, i.e.\ the degree of its lowest-order non-zero monomial, must satisfy
\beq\label{valineq} \val P_{\cG}\geq g(\cG)\, .\eeq
Crucially, as we shall explain below, this inequality may become strict due to cancellations between the contributions of the primary and descendant graphs.

It is convenient to introduce the notions of \emph{grade} $\ell(G)$ and \emph{degree} $\omega(G)$ of a stranded graph $G$, defined by
\begin{align}
\label{Model2ellExp1}&\ell(G):=4 + 2 v(G) + b(G) - f(G)- 2 \varphi(G)=2h(G)-2\alpha(\cG)\, ,\\
\label{Model2omegaExp1}&\omega(G):=3 + \frac{3}{2}v(G) + b(G) -f(G)-\varphi(G)= 2h(G)-\alpha(\cG)=h(G)+\frac{\ell(G)}{2}\,\cdotp
\end{align}
With these definitions, the power of $D$ in a stranded graph amplitude reads
\beq \label{Model2Poweromega}
\mathcal{A}(G) \propto D^{1+2h(G)-\omega(G)} =D^{1+h(G)-\frac{\ell(G)}{2}} \, .\eeq
The grade is a natural notion in view of the expansion \eqref{largeDexp} and the conjecture \eqref{conjecture}. The degree generalizes the notion of degree in \eqref{Model1ExpLargeN} or index in \cite{FRV} to the case of stranded graphs with no colored graph counterpart.

\subsection{Existence of the large $D$ expansion}

\subsubsection{Pitfalls}

Let us start by clarifying two potentially confusing points. 

The simplest scenario that would make the large $D$ limit well-defined is to have a lower bound on the grade\footnote{One could make an equivalent discussion using the degree.} of primary stranded graphs of fixed genus, of the form $\ell(G_{*})\geq 2g -2\tilde\eta(g)$ for some non-decreasing function $\tilde\eta$. The simplest possibility, suggested by the colored graph models of Section \ref{sec:complex}, is $\tilde\eta=g$. Since the contribution of descendants can only be at equal or lower order in the $1/N^{2}$ expansion, the maximum power of $D$ in any stranded graph contributing at order $N^{2-2h}$ would then be $D^{1+\tilde\eta(h)}$. This upper bound on the power of $D$, being valid term by term in the sum \eqref{Model2AmplitudeFeynman}, would automatically work for the full Feynman amplitude $\mathcal A(\cG)$ and would ensure the existence of the large $D$ limit. The limit would actually be consistent even without imposing the tracelessness condition on the matrices. However, this simple scenario works only at leading order $h=0$ \cite{Frank} but breaks down at higher order. Already at genus one, there exist stranded graphs with arbitrarily negative degree and grade, such as the chains of non-planar tadpoles shown in Figure \ref{fig:tadpole_chain}. \emph{Clearly, the large $D$ limit can work only if non-trivial cancellations occur in the sum \eqref{Model2AmplitudeFeynman}.}

\begin{figure}
    \centering
    \includegraphics[scale=.8]{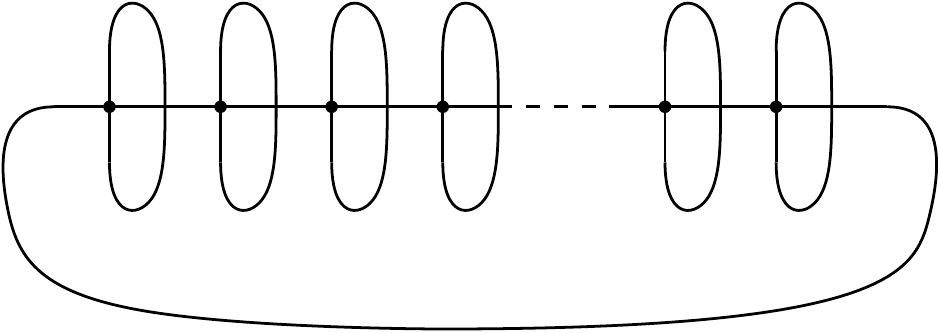}
    \caption{Chain of non-planar tadpoles with $v$ vertices associated to a genus one primary stranded graph: it has arbitrarily negative grade $\ell=2-v$ and degree $\omega=2-\frac{1}{2}v$.}
    \label{fig:tadpole_chain}
\end{figure}

Another possible expectation is that the tracelessness condition removes the whole contribution of a Feynman graph whose primary stranded graph contains a singular ribbon edge; but this is incorrect. Indeed, superficially, singular edges propagate the trace mode of the matrices, since the indices associated with their $\uN$ strands are identified. However, it is easy to check that the chain of non-planar tadpoles depicted in Fig.\ \ref{fig:tadpole_chain} is such that its primary stranded graph does contain singular edges; however, it has non-zero amplitude (coming from non-cancelled descendants). Thus, we do have to deal with graphs containing singular edges.

\subsubsection{\label{sec:lowbchi}Lower bound on $\eta(h)$}

We now derive a lower bound on the function $\eta(h)$. To do so, we prove the following proposition.

\begin{proposition} There exist Feynman graphs $\cG$ whose amplitude scales as $N^{2-2h}D^{1+h}$ at large $N$ and $D$, for all $h\in\mathbb N$.
\end{proposition}

Barring cancellations between distinct Feynman graphs, this implies that $\eta(h)\geq h$.

\begin{figure}
    \centering
    \includegraphics{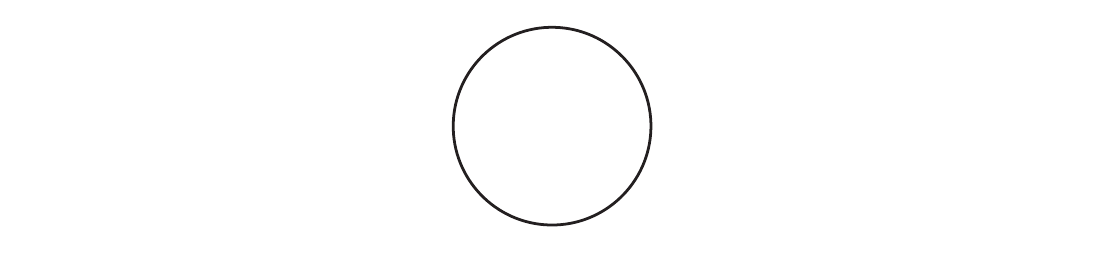}
    \caption{The ring Feynman graph.}
    \label{fig:ring}
   
\end{figure}

\begin{figure}
    \centering
    \includegraphics[scale=.8]{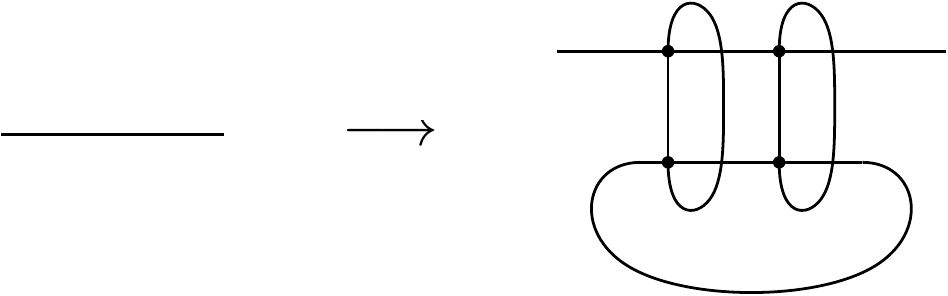}
    \caption{Move adding four vertices, two faces and three cycles to the primary graph and thus increasing the genus by one unit while preserving the grade.}
    \label{fig:moveKg}
\end{figure}

\begin{proof} Starting from the $v=0$ ring Feynman graph, see Fig.\ \ref{fig:ring}, we construct a Feynman graph $\mathcal K_{g}$, of any genus $g$, by performing $g$ times the move introduced in \cite{Azeyanagi:2017mre} and depicted in Fig.\ \ref{fig:moveKg}. It is straightforward to check that for any $g$, the associated primary stranded graph has no singular ribbon edge. Following the discussion below Eq.\ \eqref{Model2hExp2}, we find that $h(G)>g$ for any descendant $G$. The contribution of the primary graph in the amplitude $\mathcal A(\mathcal K_{g})$, which scales as $N^{2-2g}$, thus cannot be cancelled by contributions from descendants in the sum \eqref{Model2AmplitudeFeynman}. We conclude by noting that, by construction, $\alpha(\mathcal K_{g})=g$.\end{proof}

\subsubsection{Conjecture and proposition}

Our main goal is to prove that $\eta(h)$, defined in \eqref{chidef}, is finite, and in fact grows at most linearly with $h$. We do not expect that cancellations between distinct Feynman graphs could occur and help us prove the result. We thus seek a bound valid for each Feynman graph individually.

\begin{conjecture}\label{conj11} Let $\cG$ be a connected vacuum Feynman graph. Then
\beq\label{co1} \alpha(\cG)\leq \mathrm{val} \, P_{\cG}\, .\eeq
\end{conjecture}

\begin{corollary} The inequality \eqref{co1} implies \eqref{conjecture}, $\eta(h)=h$.
\end{corollary}

\begin{proof} Eq.\ \ref{modelAmplitudeFeynman2} shows that the amplitude $\mathcal A(\cG)$ scales as $N^{2-2\val P_{\cG}}$ at large $N$. The inequality \eqref{co1} then implies that $\eta(h)\leq h$. But we know that $\eta(h)\geq h$ from Section \ref{sec:lowbchi}.
\end{proof}

We shall be able to prove a weaker form of the inequality.

\begin{proposition}\label{weakprop} Let $\cG$ be a connected vacuum Feynman graph. Then
\beq\label{co2} \alpha(\cG)\leq 2\, \mathrm{val} \, P_{\cG}\, .\eeq
\end{proposition}

This inequality implies that $\eta(h)\leq 2h$ and is thus enough to prove the existence of the large $D$ limit.

\subsubsection{Reduction to the case of 2PI vacuum Feynman graphs}

To prove the inequalities \eqref{co1} or \eqref{co2}, we need to understand non-trivial cancellations between terms in the sum \eqref{Model2AmplitudeFeynman}. For instance, the chain of tadpoles $\cG_{v}$ with $v$ vertices (see Fig.\ \ref{fig:tadpole_chain}) has $\alpha = v/2$ and thus \eqref{co1} predicts that $\mathcal A(\cG_{v})$ scales at most as $N^{2-v}$ at large $N$, whereas the primary stranded graph is of genus one and thus scales as $N^{0}$ for all $v$ (actually, the chain of tadpoles scales as $N^{2-2v}$ at large $N$, see below).

Let us start by reducing the problem to the case of two-particle irreducible (2PI) Feynman graphs. This will allow in particular to fully understand the chain of tadpoles.

\begin{proposition}\label{prop2PI} Assume that 
\beq\label{ineq4}\alpha(\cG)\leq \kappa \, \mathrm{val} \, P_{\cG}\eeq
is valid for any 2PI Feynman graph $\cG$, where $\kappa$ is a positive constant (we are mainly interested in the cases $\kappa=1$ or $\kappa=2$). Then the same inequality is valid for any Feynman graph $\cG$.
\end{proposition}

\begin{proof} The proof uses a simple property of Feynman 2-point graphs.

\begin{lemma}\label{lemmaSchur} The amplitude associated to any Feynman 2-point graph is proportional to its tree-level value \eqref{Model2FreeProp}.
\end{lemma}

\begin{proof} A proof based on Schur's lemma can be given, following a similar reasoning as in Section 6, Lemma 9 of \cite{Benedetti:2017qxl}. We provide here an elementary and direct argument.

As in the case of vacuum Feynman graphs, the amplitude $\mathcal A(\tilde\cG)^{a}_{\mu b, }{}^{c}_{\nu d}$ associated to a Feynman 2-point graph $\tilde\cG$ is given by a sum over stranded graph amplitudes $\mathcal A(\tilde G_{\alpha})^{a}_{\mu b, }{}^{c}_{\nu d}$, where $\alpha$ labels the contributing stranded graphs. The strand structure, which is itself a consequence of the underlying $\uN\times\oD$ symmetry, implies that $\mathcal A(\tilde G_{\alpha})^{a}_{\mu b, }{}^{c}_{\nu d}$ is proportional to $(a_{\alpha}\delta^{a}_{d}\delta^{c}_{b}+ b_{\alpha}\delta^{a}_{b}\delta^{c}_{d})\delta_{\mu\nu}$. Summing over $\alpha$, we find that
\beq\label{AGtilde1} \mathcal A(\tilde\cG)^{a}_{\mu b, }{}^{c}_{\nu d}=\bigl(a\,\delta^{a}_{d}\delta^{c}_{b}+ b\,\delta^{a}_{b}\delta^{c}_{d}\bigr)\delta_{\mu\nu}\, ,\eeq
where $a$ and $b$ may depend on $N$, $D$ and the coupling constant $\lambda$. We now use the fact that the external strands, for instance the one associated with the indices $a$, $b$ and $\mu$, are attached to an internal vertex with the help of a propagator. This implies the following structure,
\beq\label{AGtilde2} \mathcal A(\tilde\cG)^{a}_{\mu b, }{}^{c}_{\nu d}  = 
\bP^{a}_{\mu b, }{}^{a'}_{\mu' b'} u^{\mu'b' c}_{a'\nu d}\, ,\eeq
where $u$ contains the contributions of all the other contractions in the graph. Contracting $a$ with $b$ and using $\bP^{a}_{\mu a, }{}^{a'}_{\mu' b'} =0$ implies $\mathcal A(\tilde\cG)^{a}_{\mu a, }{}^{c}_{\nu d}=0$. Comparing with \eqref{AGtilde1}, we find $a+N b=0$, which yields the desired result.
\end{proof}

\begin{figure}
    \centering
    \includegraphics[scale=1]{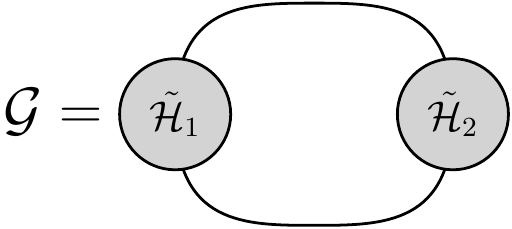}
    \caption{Structure of a connected Feynman graph $\cG$ with at least two vertices, which is not 2PI. $\tilde{\cH}_1$ and $\tilde{\cH}_2$ are connected 2-point Feynman graphs, each containing at least one vertex.}
    \label{fig:2PI}
\end{figure}

We now proceed by induction on the number of vertices in $\cG$. The amplitude for the ring Feynman graph $v=0$ is $(N^{2}-1)D$. In particular $\alpha=0$ and \eqref{ineq4} is satisfied. Note that the Feynman graphs with $v=1$ are 2PI. 

Consider now a connected Feynman graph $\cG$ with $v(\cG)\geq 2$ that is not 2PI. It must have the structure depicted in Fig.\ \ref{fig:2PI}. $\tilde{\cH}_1$ and $\tilde{\cH}_2$ are connected Feynman 2-point graphs, each containing at least one vertex. We denote by $\cH_1$ and $\cH_2$ the vacuum graphs obtained from $\tilde{\cH}_1$ and $\tilde{\cH}_2$, respectively, by closing the two external legs onto themselves. From \eqref{modelAmplitudeFeynman2}, we can write
\beq \label{tildeHvac}
\mathcal{A}(\cH_{i}) = \lambda^{v(\cH_{i})} (N^2 -1) D^{1+\alpha(\cH_{i})} P_{\cH_{i}}(N^{-2})\, ,\quad 1\leq i\leq 2\eeq
and it is straightforward to see, using Lemma \ref{lemmaSchur} and the projector property $\bP^{a}_{\mu b, }{}^{c}_{\nu d} \bP^{d}_{\nu c, }{}^{e}_{\rho f}=\bP^{a}_{\mu b, }{}^{e}_{\rho f}$, that this is equivalent to 
\beq \label{tildeH2pt}
\mathcal{A}(\tilde\cH_{i}) = \lambda^{v(\cH_{i})} D^{\alpha(\cH_{i})} P_{\cH_{i}}(N^{-2})\bP\, ,\quad 1\leq i\leq 2\, ,\eeq
for the amplitudes of the corresponding Feynman 2-point graphs. Gluing these amplitudes to get $\mathcal A(\cG)$ then yields
\beq\label{gluing}
\alpha(\cG) = \alpha(\cH_1) + \alpha(\cH_2) \quad \mathrm{and} \quad  P_{\cG} = P_{\cH_1}  P_{\cH_2} \,.\eeq
By the induction hypothesis, $\alpha(\cH_{i})\leq\kappa\val P_{\cH_{i}}$ and thus
\beq\label{endinduc1}\alpha(\cG) \leq \kappa(\val P_{\cH_{1}}+\val P_{\cH_{2}})=\kappa\val P_{\cG}\, .\eeq
\end{proof}

The above recursion compactly encodes the cancellations due to the tracelessness condition. A crucial point is the trivial-looking Lemma \ref{lemmaSchur}. A similar result would not be true without the tracelessness condition: the tree-level propagator is then $\delta^{a}_{d}\delta^{c}_{b}$ but Feynman 2-point graphs can include the other structure $\delta^{a}_{b}\delta^{c}_{d}$. More generally, a similar result could be obtained as long as the matrices (or tensors) transform under irreducible representations of the relevant symmetry groups \cite{Benedetti:2017qxl, Carrozza:MixedPerm, Carrozza:2018psc}.

Applying the recursion to the chain of tadpoles $\cG_{v}$ of Fig.\ \ref{fig:tadpole_chain} immediately yields $P_{\cG_{v}}=P_{\cG_{1}}^{v}$. From a short calculation we get $P_{\cG_{1}}(x) = x$ and thus, from \eqref{modelAmplitudeFeynman2}, the amplitude reads $\mathcal A(\cG_{v}) = \lambda^{v}(N^{2}-1)D^{1+v/2}N^{-2v}$. This result displays the many cancellations amongst the stranded graphs, which reduce the naive power counting $N^{0}$ of the primary stranded graph down to $N^{-2v}$.

The proposition \ref{prop2PI}, combined with the inequality \eqref{valineq}, has a simple but useful corollary.

\begin{corollary}
\label{corr2PI} Assume that 
\beq\label{ineq5}\alpha(G_{*})\leq \kappa g(G_{*})\eeq
is valid for any 2PI primary stranded graph $G_{*}$, where $\kappa$ is a positive constant. Then $\alpha(\cG)\leq\kappa\val P_{\cG}$ for any Feynman graph $\cG$.
\end{corollary}

This corollary, together with \eqref{Model2alphaExp1}, shows that the conjecture \ref{conj11}, or equivalently \eqref{conjecture}, would follow from a simple-to-state graph-theoretic result. Note that, in graph-theoretic terminology, our primary stranded graphs are connected 4-regular orientable embedded graphs and our $\oD$ cycles may be called ``straight-ahead'' cycles in view of the way they cross vertices.

\begin{conjecture}\label{conj2}
Let $G$ be a connected $4$-regular orientable embedded graph of genus $g$ with $v$ vertices and $\vphi$ straight-ahead cycles. If $G$ is 2PI, i.e.\ does not contain any two-edge-cut, then
\beq
\vphi \leq 1 + g + \frac{v}{2}\,\cdotp
\eeq
\end{conjecture}

We have tried to prove this conjecture inductively, using a strategy based on local modifications (called moves) of the graph akin to what we do in Section \ref{sec:hermitiansix} for the wheel interaction, but we have not succeeded yet. The main difficulty comes from the non-local constraint of two-particle irreducibility. A local move on a graph can easily produce non-trivial 2PR components. Even though we were able to overcome this problem in many instances, we were not able to fully treat stranded graphs containing $\oD$ cycles of length three (triangles). The advantage of the order six models studied in Section \ref{sec:hermitiansix} is that triangles do not arise. Another way to eliminate triangles is to consider the bipartite version of the model, as we do in \ref{sec:bitetra}.

\subsubsection{Proof of Proposition \ref{weakprop}}

Instead of \eqref{co1}, let us thus focus on the weaker inequality \eqref{co2}. We are actually going to prove the inequality \eqref{ineq5} for $\kappa=2$. It follows rather easily from the results of the previous subsection combined with a highly non-trivial general power-counting bound proven for a suitable family of stranded graphs in \cite{Benedetti:2017qxl}.

\begin{proposition}\label{propo1} \cite{Benedetti:2017qxl}
Let $G$ be a connected stranded graph. If $G$ does not contain elementary melon or tadpole $2$-point subgraphs (see Fig.\ \ref{fig:tadpolemelon}), then $\omega(G) \geq 0$ (or, equivalently, $\alpha(G)\leq 2h(G)$). 
\end{proposition}

\noindent\emph{Proof of Prop.\ \ref{weakprop}}. Primary stranded graphs with one or two vertices are obtained by gluing the external legs of the nine Feynman 2-point graphs depicted in Fig.\ \ref{fig:tadpolemelon}. For all these cases, it is immediate to check that $\alpha\leq 0$ and thus \eqref{ineq5} is satisfied. A 2PI primary stranded graph with $v\geq 3$ clearly cannot contain an elementary melon or tadpole $2$-point subgraph. Since $h=g$ for primary stranded graphs, we can apply Prop.\ \ref{propo1} to derive \eqref{ineq5} for $\kappa=2$. Prop.\ \ref{weakprop} then follows from Corr.\ \ref{corr2PI}.\qed



\section{\label{sec:hermitiansix}Hermitian tetrahedral bipartite and wheel interactions}

In this section, we focus on models for which the existence of the large $D$ limit and the bound \eqref{conjecture} do not require to impose the tracelessness condition. In these cases, the free propagator that replaces \eqref{Model2FreeProp} is simply
\beq \label{ModelBipartiteGen}
\bigl\langle X^{a}_{\mu\, b} X^{c}_{\nu\, d} \bigr\rangle_0 = \delta_{\mu\nu}\delta^a_d \delta^c_b\, .\eeq
There is now a unique stranded graph associated with a Feynman graph and we denote both by $G$. Most of the discussion in Section \ref{sec:ganda} is then much simplified. 

\subsection{\label{sec:gensecherm}Amplitudes, $n$-cycles and strategy of proof}

Let us define the \emph{grade} of a general stranded graph, containing $e$ propagators \eqref{ModelBipartiteGen} and $v$ vertices associated with bubbles $\cB_{a}$ having $x(\cB_{a})=x_{a}$ (see Section \ref{reviewFRV}) as 
\begin{align}\label{levelbip1}\ell & :=4c - 2v+2e-f-2\varphi+\sum_{a}x_{a}\\
\label{levelbip2} & = 2c+2g+e+\sum_{a}x_{a}-v-2\varphi\, ,
\end{align}
where $c$ is the number of connected components of the graph and $g$ its genus. This formula replaces and is consistent with \eqref{Model2ellExp1}; note that presently $b=0$ and that for the quartic interaction studied in Section \ref{sec:Hermitian}, $e=2v$ and $x=0$.

Using \eqref{GenusColored2}, with $x(\cB_{a})=x_{a}$, and the fact that, in a connected graph containing interactions of various orders $2s_{a}$, the number of edges is given by
\beq\label{edgeformula} e=\sum_{a}s_{a}\, ,\eeq
we find that the amplitude of a general graph can be written as
\beq \label{ModelBipartiteAmpl}
\cA= |\lambda|^{v} N^{2-2g}D^{1+g-\frac{\ell}{2}}\, .\eeq
Assuming that we can find graphs having $\ell=0$ at any genus, which we can for the models we study below, the strong form of the conjecture given by \eqref{conjecture} is thus equivalent to
\beq\label{lpositive}\ell\geq 0\, .\eeq

We call \emph{$n$-cycles} the $\oD$ cycles of length $n$, i.e.\ passing through $n$ edges. If $\varphi_{n}$ is the number of $n$-cycles, then
\beq\label{looprel} \varphi = \sum_{n\geq 1}\varphi_{n}\,,\quad e=\sum_{n\geq 1}n\varphi_{n}\, .\eeq
If $s_{0}\geq 2$ is the minimal value taken by the $s_{a}$, i.e.\ $2s_{0}$ is the lowest possible degree for an interaction, \eqref{edgeformula} implies that $e\geq s_{0}v$. Since $x_{a}\geq 0$, Eq.\ \eqref{levelbip2} then yields
\beq\label{ellineqa} \ell\geq 2+2 g + \sum_{n\geq 1}\frac{(s_{0}-1)n-2s_{0}}{s_{0}}\,
\varphi_{n}\, .\eeq
This inequality is very useful. It implies that:
\begin{itemize}

\item If the quartic tetrahedral interaction is included, $s_{0}=2$; then graphs for which $\varphi_{1}=\varphi_{2}=\varphi_{3}=0$ have $\ell\geq 2+2g$.

\item If the model has only sextic or higher interaction terms, $s_{0}\geq 3$; then graphs for which $\varphi_{1}=\varphi_{2}=0$ have $\ell\geq 2+2g$.

\end{itemize}

These simple facts suggest a straightforward strategy to prove that $\ell\geq 0$, which is standard in the tensor models literature (see e.g. \cite{Bonzom_2011}). We proceed recursively on the number $v$ of vertices in the graph. It is convenient to work with graphs that are not necessarily connected. It is trivial to check that $\ell\geq 0$ for the ring graph $v=0$ depicted in Fig.\ \ref{fig:ring} (it is also an easy but instructive exercice to check this explicitly for a few other graphs, for instance all the $v=1$ graphs). We then consider a graph with an arbitrary number $v$ of vertices. If this graph does not contain 1-cycles, 2-cycles or 3-cycles, we are done. If it does, we devise local modifications of the graph, called \emph{moves}, that reduce the number of vertices and at the same time eliminate the unwanted $n$-cycles, \emph{without increasing $\ell$.} We can then conclude by the recursion hypothesis.

It turns out that the configurations with 3-cycles are the trickiest. These configurations can be made irrelevant either by considering a models with a bipartite structure, for which cycles must have even length, as we do in \ref{sec:bitetra}; or by considering models with $s_{0}\geq 3$, as we do in \ref{sec:Hermwheel}.

\smallskip

\noindent\emph{Remark: normal-ordered and bipartite models with no MST interactions}

In a general model with no MST interaction terms, $x_{a}\geq 1$ and thus $\sum_{a}x_{a}\geq v$. Plugging this inequality in \eqref{levelbip2} and using \eqref{looprel}, we find that for a connected graph
\beq\label{noMSTineq} \ell\geq 2+2 g +\sum_{n\geq 1}(n-2)\varphi_{n}\, .\eeq
In particular, in normal-ordered models, for which self-contractions are forbidden, or in bipartite models, for which odd-length cycles are forbidden, the existence of the large $D$ limit is trivially ensured if no MST terms are included, with $\eta(h)\leq -1$.

\subsection{\label{sec:bitetra} The complex bipartite tetrahedral model}

We study the complex bipartite tetrahedral model, which is the model of the form \eqref{ModelBipartiteColored} with the tetrahedral interaction. One has $e=2v$ and the grade \eqref{levelbip1} reads
\beq\label{levtetbip} \ell = 4 c + 2v -f-2\varphi\, .\eeq
Feynman graphs in this model are \emph{bipartite}: black vertices are chosen to correspond to interactions $I_{\cB_T}(X,X)$ in \eqref{ModelBipartiteColored}, white vertices to their Hermitian conjugates and edges can join black and white vertices only. The bipartite structure of the graph implies that cycles must have even length. Graphs that do not contain 2-cycles thus have $\ell\geq 2+2g$.

\begin{proposition} \label{PropBipartite4} In the complex bipartite tetrahedral model, any Feynman graph $G$ has non-negative grade, $\ell\geq 0$ and this is the best possible bound.
\end{proposition}

\begin{figure}
    \centering
    \includegraphics[scale=1]{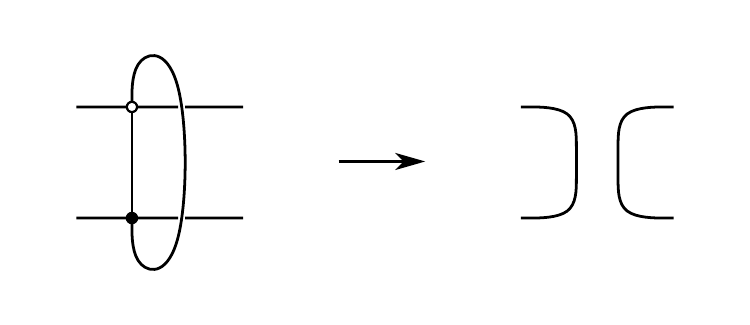}
    \caption{Dipole deletion move.}
    \label{fig:dipoledeletion}
\end{figure}

\begin{figure}
    \centering
    \includegraphics[scale=.7]{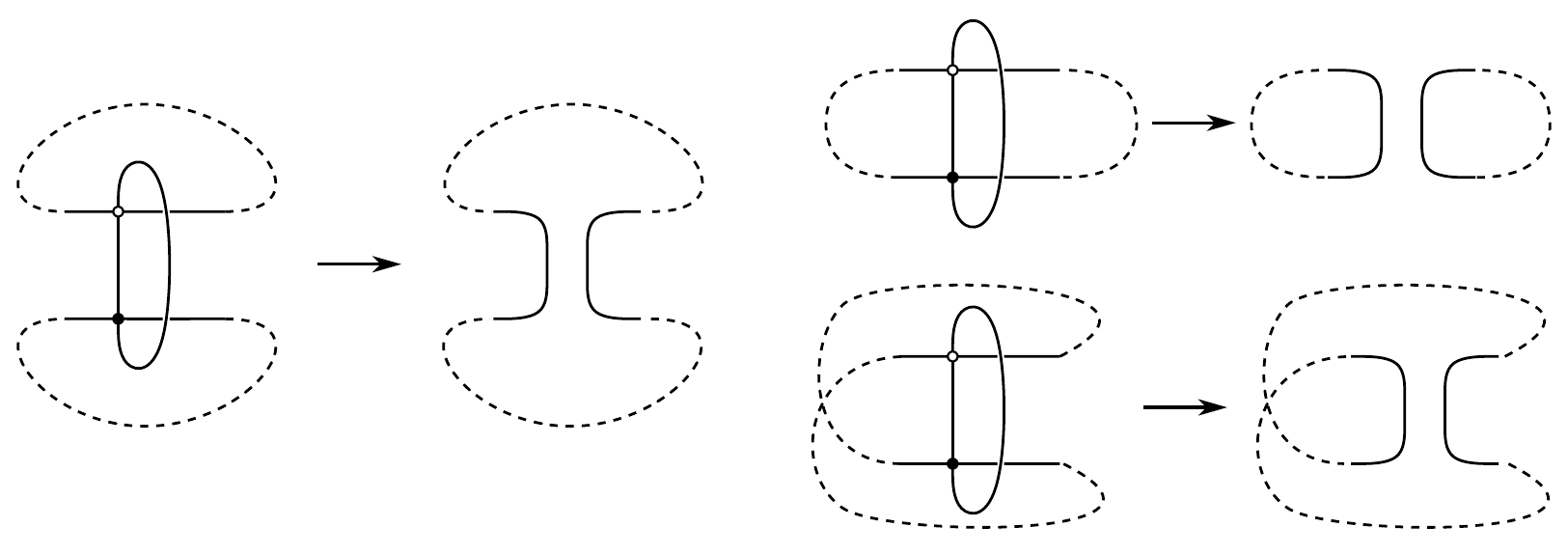}
    \caption{The three possible cycle configurations discussed in the main text.}
    \label{fig:cycleindipoletetra}
\end{figure}

\begin{proof} As outlined in \ref{sec:gensecherm}, we proceed by induction on $v$. 

If $v = 0$, $G$ is a collection of ring graphs; hence, $f = 2 c = 2 \vphi$ and $\ell = 0$.

Let $v>0$ and assume that all graphs, connected or not, with strictly less vertices than $v$, have $\ell\geq 0$. Consider a graph $G$ with $v$ vertices, $f$ faces, $\varphi$ cycles and $c$ connected components. If it does not contain a 2-cycle, $\ell\geq 0$ and we are done. If it does contain a 2-cycle, then, taking into account the bipartite contraint, the non-planar dipole represented on the left of Fig.\ \ref{fig:dipoledeletion} must be present in the graph.

We then consider the local move (``dipole deletion'') depicted in Fig.\ \ref{fig:dipoledeletion}. This yields a new graph $G'$ with $v'$ vertices, $f'$ faces, $\varphi'$ cycles and $c'$ connected components. The move preserves the bipartite character of the graph and the face structure, so that $f'=f$. It removes two vertices, so that $v'=v-2$. Moreover, the structure of the move implies that the number of connected component either remains the same, $c'=c$, or increases by one unit, $c'=c+1$.

One can then contemplate the three possible situations for the cycles that go through the dipole, Fig.\ \ref{fig:cycleindipoletetra}. If we have three distinct cycles to begin with (left inset of the figure), then $c'=c$ and $\varphi'=\varphi-2$, which yields $\ell'=\ell$. If we have two distinct cycles to begin with, then: either $\varphi'=\varphi$ (upper-right inset) which, taking into account $c'-c\leq 1$, yields $\ell'\leq\ell$; or $c'=c$ (lower-right inset) and then $\varphi'=\varphi-1$ and $\ell'=\ell-2$. In all cases, $\ell'\leq \ell$ and thus $\ell\geq 0$ by the induction hypothesis.

Note finally that the bound $\ell\geq 0$ is the best possible. At genus zero, the usual (genera\-lized) melons have $\ell=0$. One can then use the move depicted in Fig.\ \ref{fig:moveKg}, which respects the bipartite structure, to construct $\ell=0$ graphs at any genus. We conclude that the conjecture \ref{conjecture}, $\eta(h)=h$, holds.\end{proof}

\subsection{\label{sec:Hermwheel}The Hermitian wheel model}

We consider the Hermitian wheel model, for which the action reads
\beq \label{HermitianWheel}
S = \frac{1}{2}\tr X_\mu X_\mu + \frac{\lambda}{8 N^{2}D} \, \tr X_\mu X_\nu X_\rho X_\mu X_\nu X_\rho \, .\eeq
The graphs of this model have $e=3v$ and the grade \eqref{levelbip1} reads
\beq\label{levwheel} \ell = 4 c + 4v -f-2\varphi\, .\eeq
Note also that Euler's formula $2c-2g=f-e+v=f-2v$ implies that the number of faces is necessarily even. 

\begin{proposition} \label{PropHermitianWheel} In the Hermitian wheel model, any Feynman graph $G$ has non-negative grade, $\ell\geq 0$ and this is the best possible bound.
\end{proposition}

\begin{proof} We proceed again by induction on $v$, as outlined in \ref{sec:gensecherm} and as in the proof of Prop.~\ref{PropBipartite4}.

If $v=0$, $G$ is a collection of ring graphs and has $\ell=0$.

Let $v>0$ and assume that all graphs, connected or not, with strictly less vertices than $v$, have $\ell\geq 0$. Consider a graph $G$ with $v$ vertices, $f$ faces, $\varphi$ cycles and $c$ connected components. If it does not contain a 1-cycle or a 2-cycle, $\ell\geq 0$ and we are done.

In the following, we consider local moves on the graph $G$. We note the variations of the various attributes of the graph as $\Delta v$, $\Delta f$, etc.

\paragraph{One-cycles}

\begin{figure}[h]
    \centering
    \includegraphics[scale=1.2]{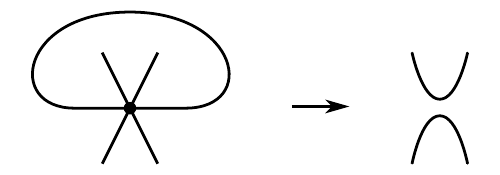}
    \caption{Move deleting 1-cycles.}
    \label{fig:p1}
\end{figure}

\begin{figure}[h]
    \centering
    \includegraphics[scale=1.2]{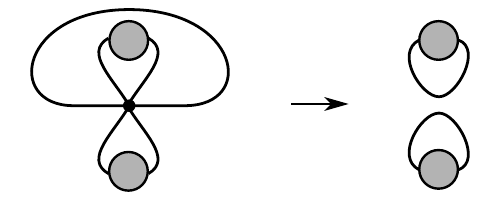}
    \caption{Case for which the move yields $\Delta c=1$. Grey shaded circles represent connected 2-point graphs.}
    \label{fig:p2}
\end{figure}

If $G$ contains a 1-cycle, a vertex with a self-contraction, as represented on the left of Fig.\ \ref{fig:p1}, must be present in the graph. We then consider the local move depicted in the figure. Obviously, $\Delta v = -1$. The move preserves the structure of the face and thus $\Delta f =0$. Initially, three cycles, at most, may go through the vertex, thus $\Delta\varphi\geq -2$. Morevoer, $\Delta c=0$ or 1. If $\Delta c=0$, $\Delta \ell = -4-2\Delta\varphi\leq 0$. If $\Delta c =1$, the initial structure must be as depicted in Fig.\ \ref{fig:p2}. But then there are only two distinct cycles that go through the vertex in the initial configuration, and these yield two distinct cycles after the move is performed. Thus $\Delta\varphi =0$ and $\Delta \ell = 0$. In all cases, $\Delta\ell\leq 0$ and thus $\ell\geq 0$ by the induction hypothesis.

\paragraph{Two-cycles that go through a unique vertex}

\begin{figure}[h]
    \centering
\includegraphics[scale=1]{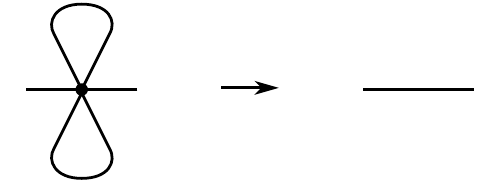}\hskip 2.5cm \includegraphics[scale=1]{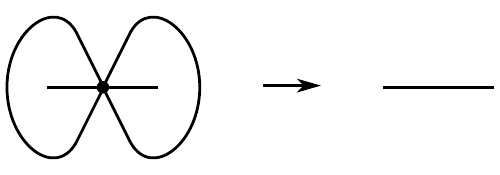}
    \caption{Move deleting 2-cycles that go through a unique vertex.}
    \label{fig:p3}
\end{figure}

If $G$ contains a 2-cycle, the 2-cycle may go through a unique vertex or through two distinct vertices. The first case is the simplest and yields two possible configurations depicted in Fig.\ \ref{fig:p3} on which we apply simple deletion moves. For the case on the left of the figure, the move yields $\Delta v=-1$, $\Delta c=0$, $\Delta f=-2$ and $\Delta\varphi=-1$, so $\Delta\ell = 0$. For the case on the right of the figure, $\Delta v=-1$, $\Delta c=0$, $\Delta f\geq -2$ (because there are at most three distinct faces that go through the vertex to start with, and there is of course at least one remaining after the move) and $\Delta \varphi = -1$. Overall, $\Delta\ell\leq 0$. In all cases, $\Delta\ell\leq 0$ and thus $\ell\geq 0$ by the induction hypothesis.

\paragraph{Two-cycles that go through two distinct vertices}

\begin{figure}[h]
    \centering
\includegraphics[scale=.8]{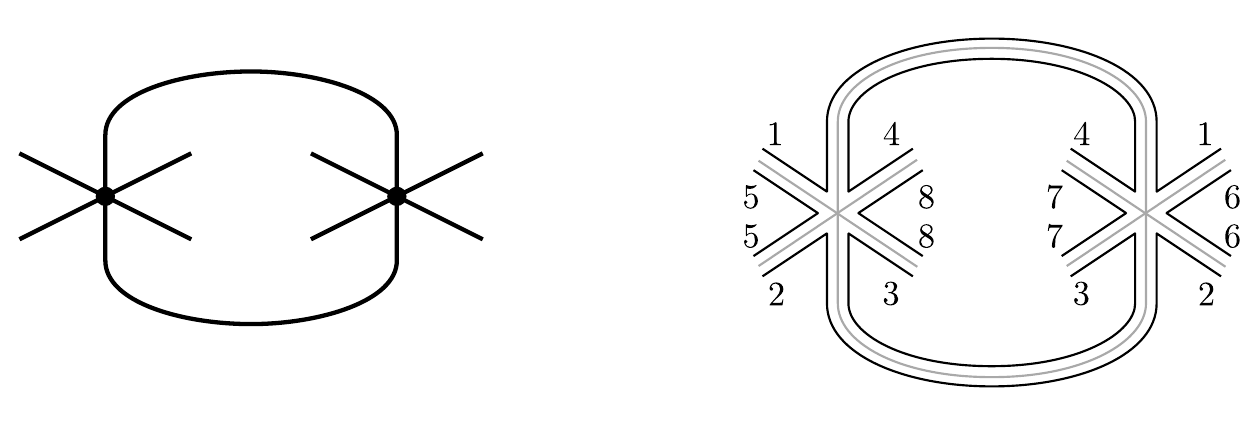}
    \caption{The configuration for which a 2-cycle goes through two distinct vertices. Left inset: Feynman graph. Right inset: stranded graph with labels used in the discussion of the main text for the faces that go through the vertices. Note that distinct labels may correspond to the same face.}
\label{fig:p4}
\end{figure}

\begin{figure}[h]
    \centering
\includegraphics[scale=.7]{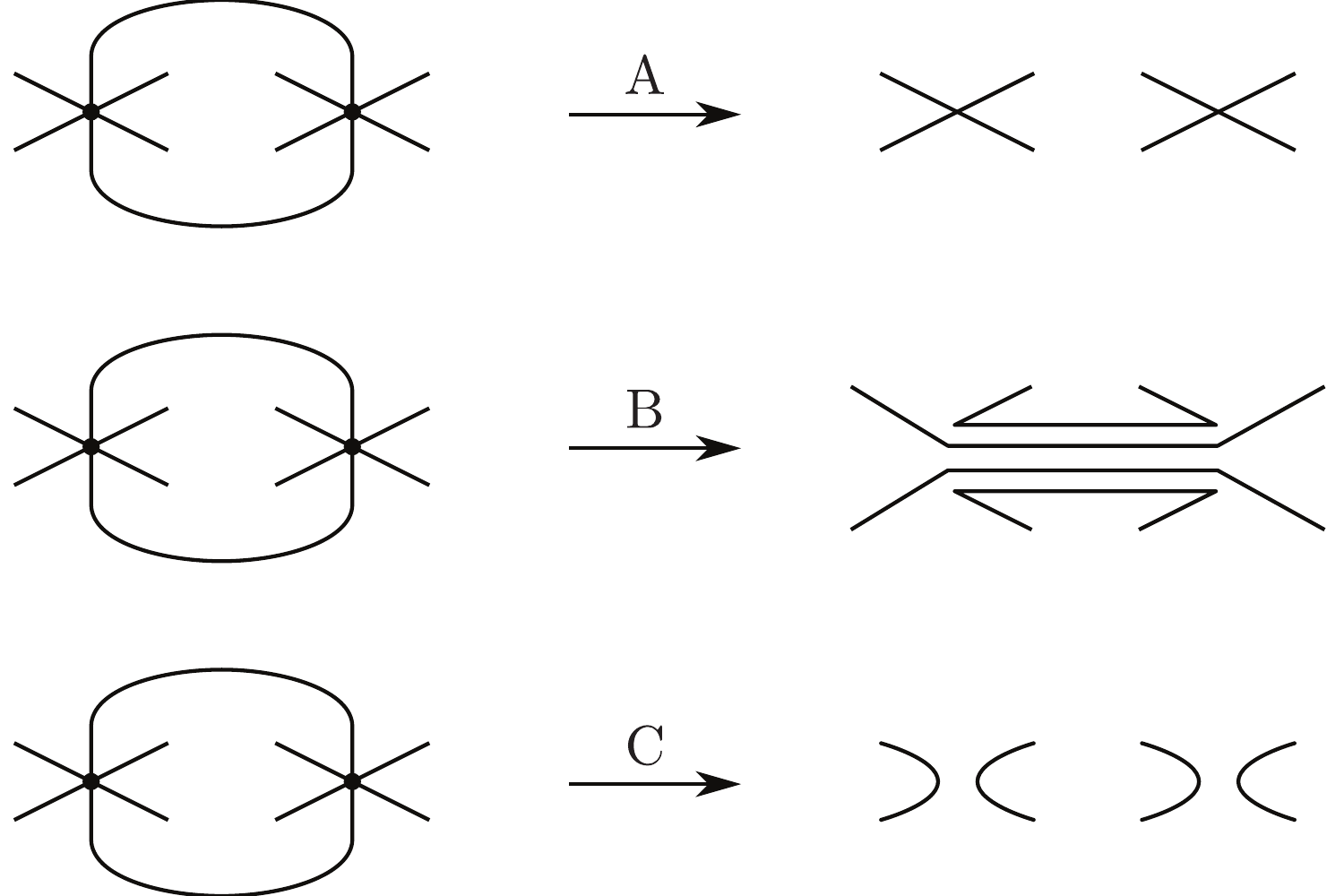}
\caption{Moves A, B and C deleting 2-cycles that go through two distinct vertices.}
\label{fig:p5}
\end{figure}

Finally, if $G$ contains a 2-cycle that go through two distinct vertices, we are in the situation depicted in Fig.\ \ref{fig:p4}. In the rest of the proof, when we refer to faces or cycles, we always mean faces or cycles that go through the two vertices. There are thus at most 5 cycles and 8 faces; the latter are numbered 1 to 8. Of course, depending on the global structure of the graph, distinct numbers may be associated to the same face. 

To treat this case, we consider three possible $\Delta v = -2$ moves depicted in Fig.\ \ref{fig:p5}. It is useful to recall that the number of faces in our graphs must be even and thus $\Delta f$ must also be even. It is also clear that $0\leq\Delta c\leq 3$.
We have the following general properties.
\begin{itemize}
\item
The move A preserves the structure of the cycles, except of course for the 2-cycle that we delete. 
Starting from the eight faces, no more than seven of them can be deleted. Furthermore, this number must be even, which brings the bound down to six. We thus have $\Delta f\geq -6$, $\Delta\varphi = -1$ and $\Delta\ell = 4\Delta c - \Delta f - 6\leq 4\Delta c$. 
\item The moves B and C preserve the structure of four of the faces (1, 2, 3, 4 for move B and 5, 6, 7, 8 for move C) and can be achieved by two cut-and-glue operations: on the pairs of strands $(5, 6)$ and $(7,8)$ (move B), or $(1, 2)$ and $(3, 4)$ (move C). 
As for the cycles, a priori we may start with five distinct cycles and end up with only one. However, this does not occur: applying the cut-and-glue operations associated with either move B or C yield a configuration that always has two cycles left, see Fig.\ \ref{fig:p6}. We thus have $\Delta f\geq -2$, $\Delta\varphi\geq -3$ and $\Delta\ell\leq 4\Delta c$.
\end{itemize}

\begin{figure}[h]
    \centering
\includegraphics[scale=.7]{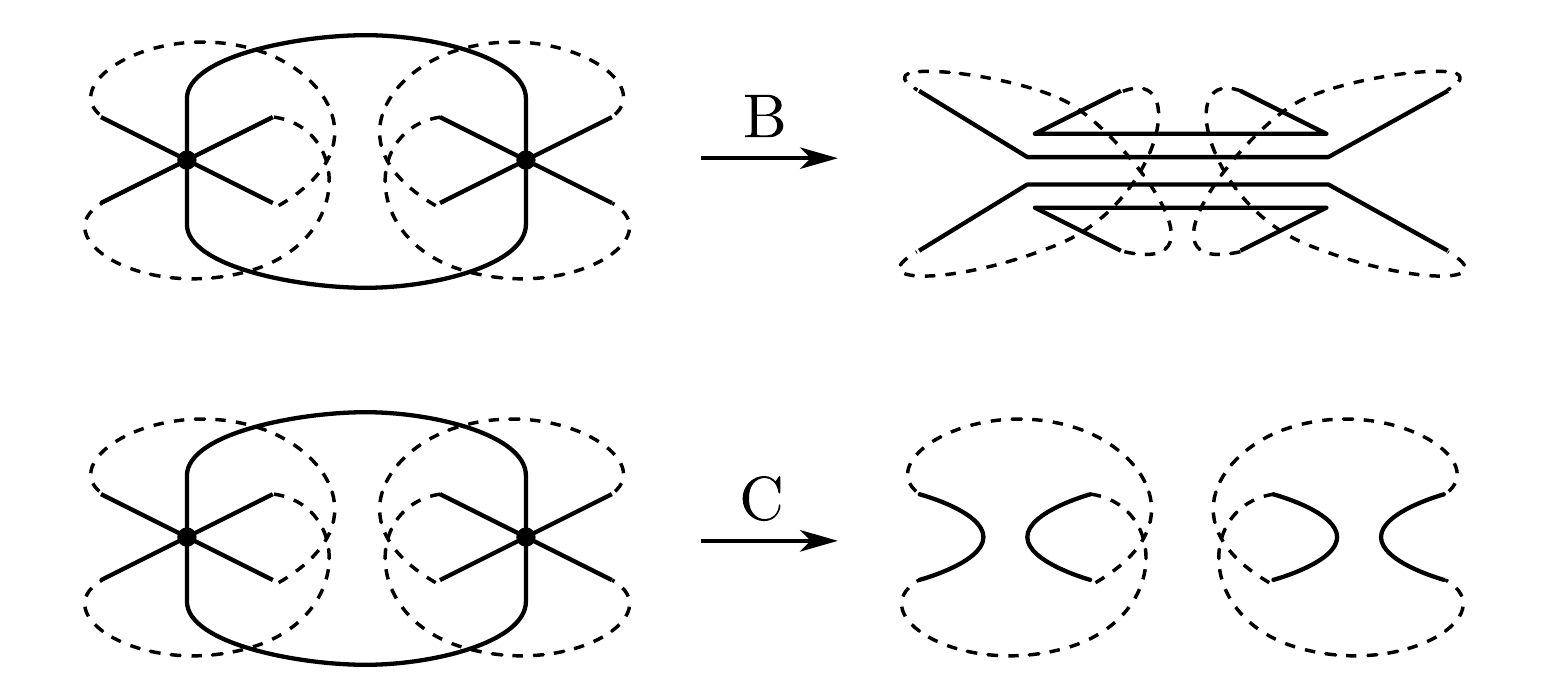}
\caption{Starting with five distinct cycles passing through the vertices and applying the moves B or C yields two distinct cycles and thus $\Delta\varphi = -3$.}
\label{fig:p6}
\end{figure}

In particular, we note that $\Delta\ell\leq 4\Delta c$ for the three moves.

Let us start by using the move B and see what happens.
\begin{itemize}
\item If $\Delta c = 0$, $\Delta\ell\leq 0$ and we're done.
\begin{figure}[h]
    \centering
\includegraphics[scale=1]{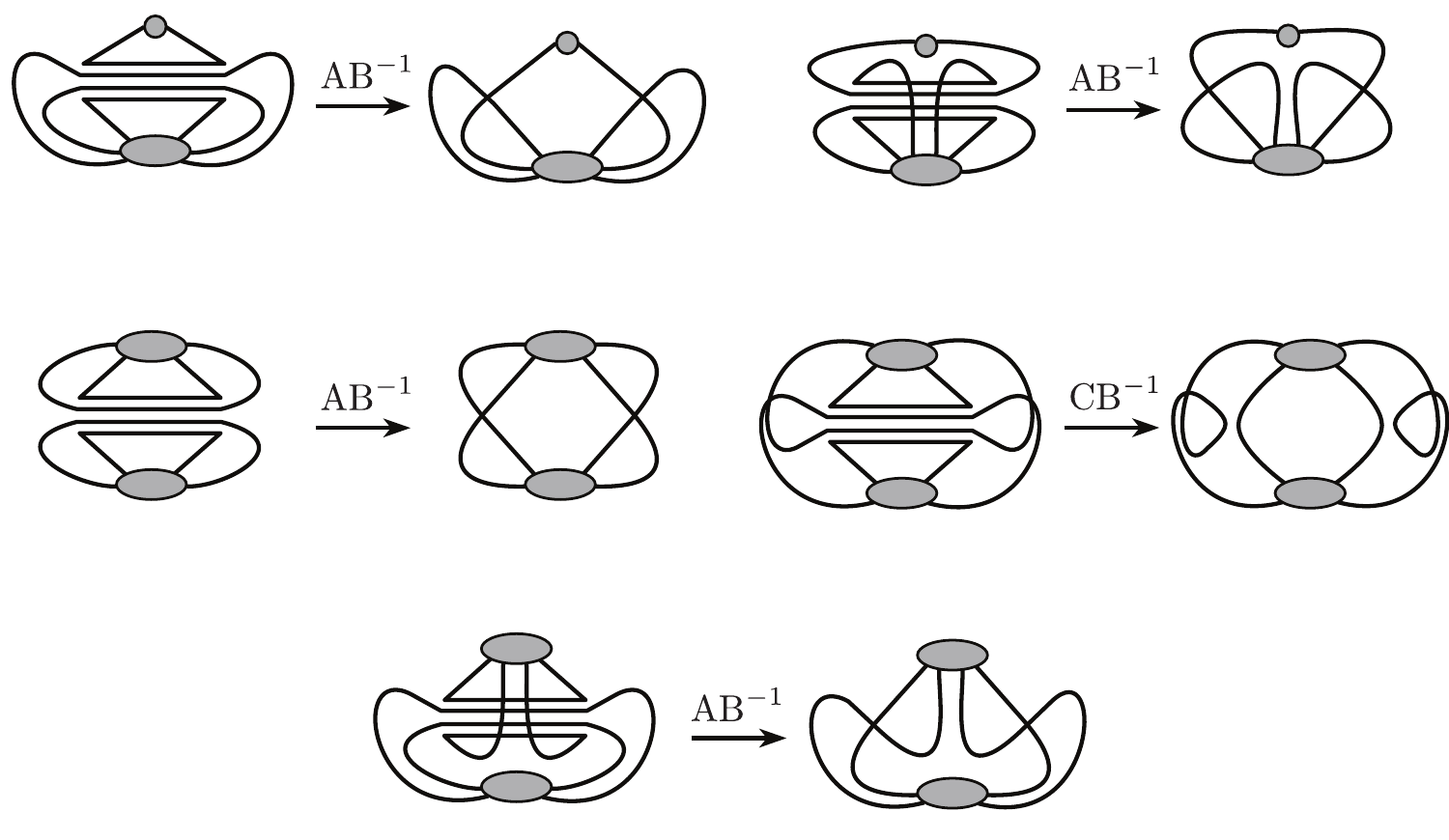}
    \caption{Configurations obtained from the move B when $\Delta c =1$ (left graph of each inset) can be replaced by new configurations for which $\Delta c =0$ by using either A or C. We denote these transformations by $\text{A}\text{B}^{-1}$ or $\text{C}\text{B}^{-1}$, to emphasize the fact that we first cancel the B-move and replace it by an A- or C-move instead. Grey shaded circles represent connected 2-point, 4-point or 6-point graphs.\label{fig:p7}}
\end{figure}
\item If $\Delta c = 1$ then, up to a trivial upside/down symmetry of the configurations, we are in one of the five cases depicted in Fig.\ \ref{fig:p7} (where the grey blobs represent connected subgraphs). As shown in the figure, we can then apply either move A or C instead of B to produce a new graph with the same number of connected components as the original and thus such that $\Delta\ell\leq 0$.
\begin{figure}[h]
    \centering
\includegraphics[scale=1]{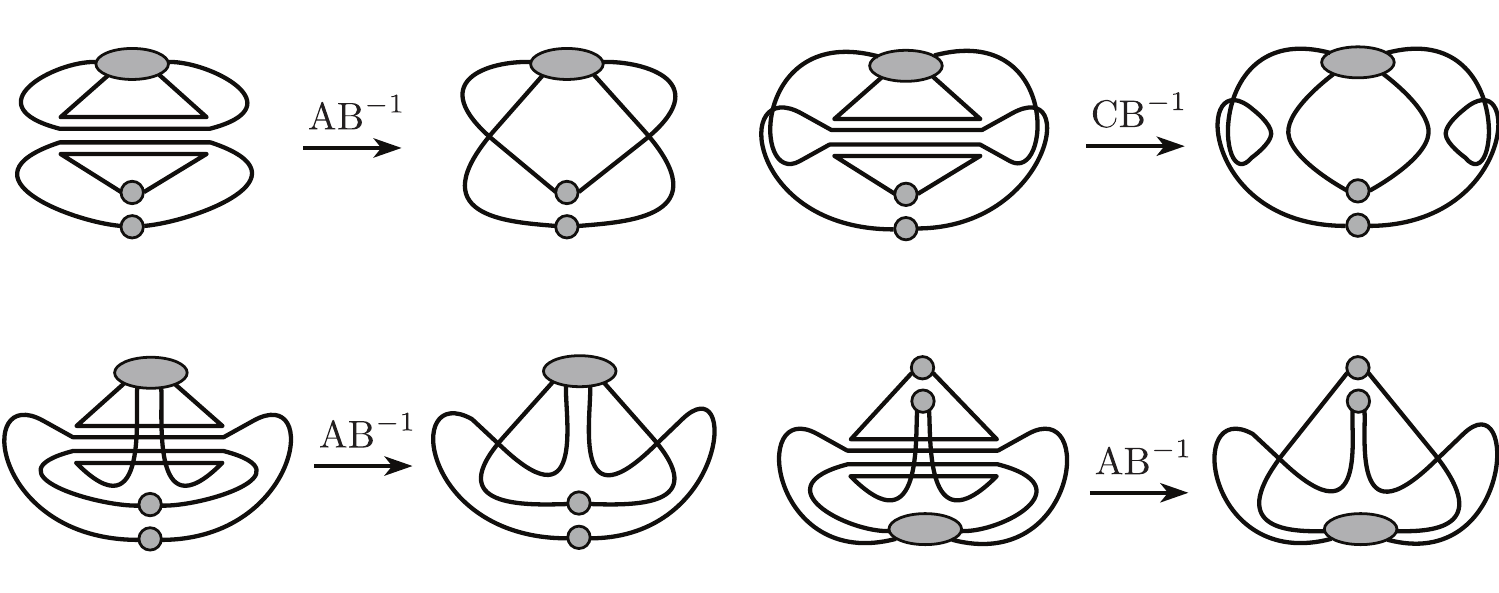}
    \caption{Similarly to Fig.\ \ref{fig:p7}, configurations obtained from the move B when $\Delta c =2$ (left graph of each inset) can be replaced by new configurations for which $\Delta c =0$ by using either A or C.\label{fig:p8}}
\end{figure}
\item If $\Delta c = 2$ then, up to a trivial upside/down symmetry of the configurations, we are in one of the four cases depicted in Fig.\ \ref{fig:p8}. As shown in the figure, we can then apply either move A or C instead of B to produce a new graph with the same number of connected components as the original and thus such that $\Delta\ell\leq 0$.
\begin{figure}[h]
    \centering
\includegraphics[scale=1]{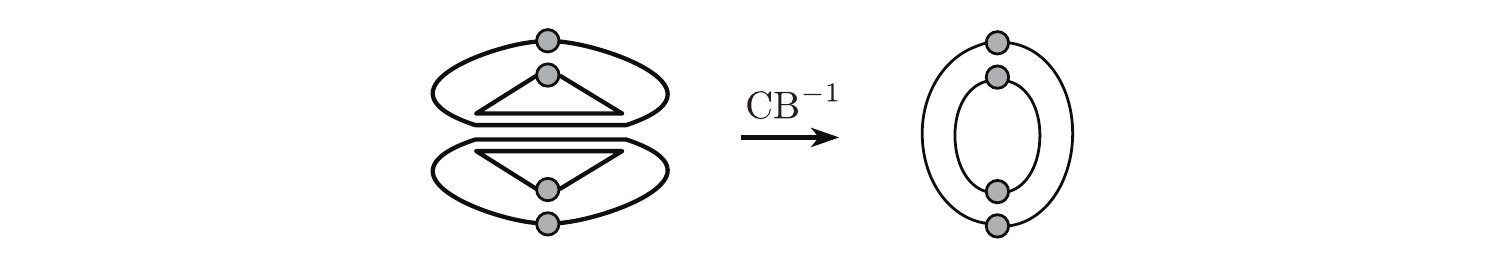}
    \caption{The configuration obtained from the move B when $\Delta c =3$ (left graph) can be replaced by a new configuration by using the move C instead, for which $\Delta c =1$ and $\Delta\varphi = -1$.\label{fig:p9}}
\end{figure}
\item Finally, if $\Delta c = 3$, the structure of the external legs clearly implies that we have three distinct cycles in the initial configuration. Moreover, from Fig.\ \ref{fig:p9} we see that if we apply the move C, we get two distinct cycles and two connected components at the end, thus $\Delta\varphi = -1$ and $\Delta c = 1$. But then, using $\Delta f\geq -2$, we get $\Delta\ell\leq 4-8+2+2=0$ and we conclude.
\end{itemize}

\begin{figure}[h]
    \centering
    \includegraphics[scale=.8]{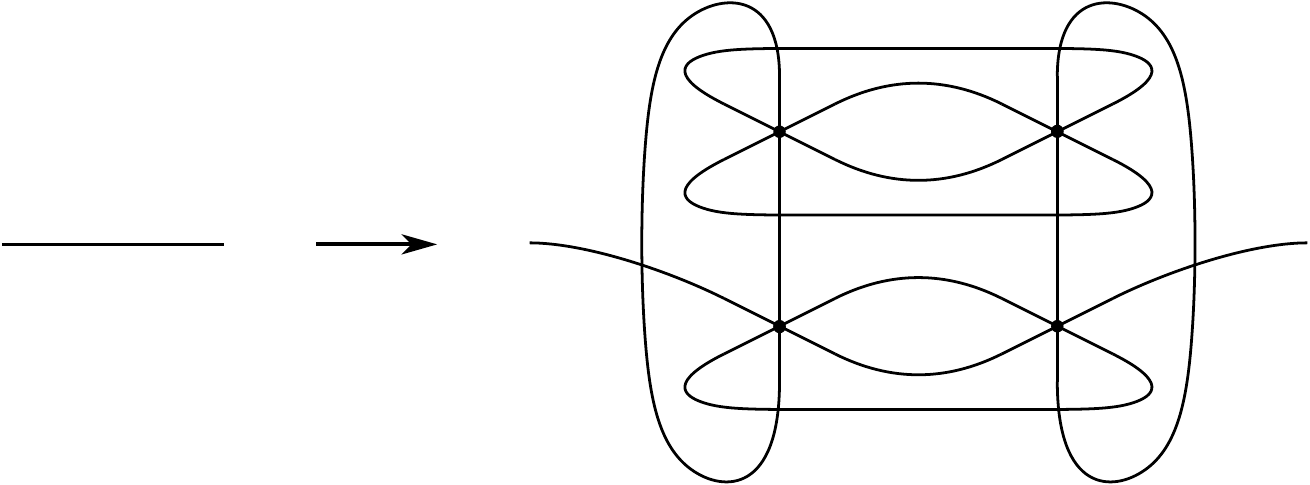}
    \caption{Move for which $\Delta v = 4$, $\Delta f=6$, $\Delta\varphi=5$ and thus $\Delta g=1$ and $\Delta\ell=0$. We remark that, the graph on the right being bipartite, $\ell=0$ graphs of arbitrary genus also exist in the bipartite wheel model.}
    \label{fig:move-increase-l-wheel}
\end{figure}

We have proven that $\ell\geq 0$. To prove that this is the best possible bound, we first note that the mirror melon construction in \cite{FRV} provides graphs with $\ell=0$ at $g=0$ (see also \cite{bonzom2015colored, prakash2019melonic}). We then use the move depicted in Fig.\ \ref{fig:move-increase-l-wheel} to generate $\ell=0$ graphs at any genus.\end{proof}

\section{Concluding remarks and perspectives}

There are several directions in which the work in the present paper could be extended. 

We first note that an infinite number of interactions supporting a well-defined large $D$ expansion can be generated straightforwardly from the two building blocks we have investigated (the tetrahedral and wheel interactions). Indeed, the boundary graph of any Feynman graph can itself serve as effective interaction\footnote{Such an effective interaction may be multi-trace.}, and with the appropriate scaling, is guaranteed to lead to a well-defined large $D$ expansion. As an illustration, one may consider an action of the form
\beq\label{eq:tree1}
S = \text{tr} \Bigl(X_\mu^\dagger X_\mu + \frac{1}{12 (N \sqrt{D})^{5}} \bigl( \lambda\,  X_{\mu_1} X_\nu X_{\mu_1} X_{\mu_2} X_\rho X_{\mu_2} X_{\mu_3} X_\nu X_{\mu_3} X_{\mu_4} X_\rho X_{\mu_4} + \mathrm{H.c.} \; \bigr) \Bigr) \, ,
\eeq
whose interaction term is represented in Fig. \ref{fig:tree1}, and has the same boundary structure as a bipartite tree of tetrahedral interactions. As a result, the action \eqref{eq:tree1} has a well-defined large $D$ expansion, which can be investigated with the methods of the present paper. 

Beyond such immediate generalizations, the technically simpler class of bipartite models \eqref{ModelBipartiteColored} is certainly worth exploring. In particular, it would be interesting to determine whether, in this context, all MST interactions support a large $D$ expansion, even in the absence of tracelessness condition (as the tetrahedral and wheel interactions do).

More ambitious, but also more physically relevant, is the study of Hermitian (or real symmetric and antisymmetric) models. In this case, the tracelessness condition is probably necessary in most examples. For instance, without it, the prism \eqref{prismdef}, and any other effective interaction decomposable into a tree of tetrahedra as in Fig.~\ref{fig:tree1}, would generate unbounded contributions through graphs analogous to the chains of tadpoles of Fig.~\ref{fig:tadpole_chain}. It would nonetheless be interesting to characterize the class of interactions, of which the wheel is one example, that can support a large $D$ expansion even in the absence of the tracelessness condition. In the longer term, one would like to obtain a full understanding of all the MST models and eventually address the most general models discussed in Section \ref{sec:conj}, for which several interaction terms are switched on at the same time. Note that it seems plausible that all MST interactions might be decomposable into trees of elementary lower-order irreducible building blocks, including, for instance, the tetrahedron and wheel. If true, the construction of a complete set of such irreducible interactions would provide a promising way forward. 

\begin{figure}[h]
    \centering
    \includegraphics{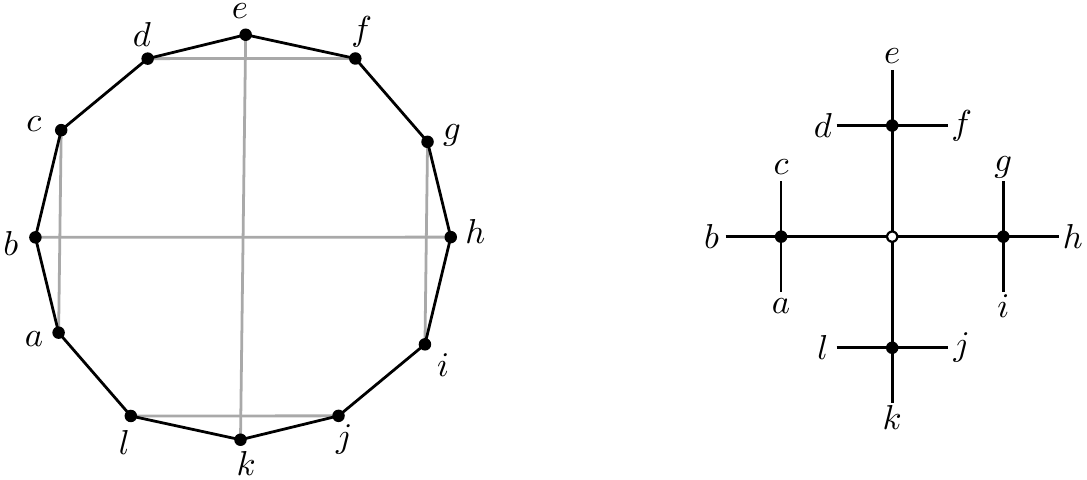}
    \caption{Interaction of order $12$ represented as a chord diagram (left): it is equivalent to a bipartite tree of tetrahedral interactions (right).}
    \label{fig:tree1}
\end{figure}

Another outstanding research direction is the problem of the classification of leading graphs. A first step is to find the optimal bounds, i.e.\ the precise form of the function $\eta$ defined in \eqref{chidef}. At planar order, we can always use the framework of $\uN_{\text{L}}\times\uN_{\text{R}}$ invariant models \cite{Frank}, for which the colored graph technology can help, but even in this case classifying the leading graphs remains difficult and only a few cases have been treated rigorously, see e.g.\ \cite{FRV}. Recent progress has been achieved in the context of the $\uN_{\text{L}}\times\uN_{\text{R}}$ invariant tetrahedral model, where $\ell=0$ graphs of \emph{arbitrary genus} have been characterized \cite{double-scaling-largeD}. One would like to generalize this classification to the case of the Hermitian model too. From this, one can obtain non-trivial physical informations to all order in the $1/N^{2}$ expansion. The $\ell=0$ graphs at any genus also yield the leading contribution in an interesting new double-scaling limit where $N$ and $D$ are sent to infinity while keeping $N^2/D$ finite.

Finally, we note that while the proofs we have presented have a general logic, the details that work for a particular model are not automatically generalizable to other cases, even when they superficially look similar. To address this shortcoming, we hope at term to develop a general formalism that will be able to supplant the colored combinatorics used in \cite{FRV} in the broader context we have started to explore in the present paper.

\subsection*{Acknowledgments}

SC would like to thank ULB for hospitality in the early stages of this collaboration, and Valentin Bonzom for stimulating discussions along the way. GV is grateful for the hospitality of Perimeter Institute where part of this work was carried out. GV is a Research Fellow at the Belgian F.R.S.-FNRS.

Research at Perimeter Institute is supported in part by the Government of Canada through the Department of Innovation, Science and Economic Development Canada and by the Province of Ontario through the Ministry of Colleges and Universities.

This research is supported in part by the Belgian Fonds National de la Recherche Scientifique FNRS (convention IISN 4.4503.15) and the F\'ed\'eration Wallonie-Bruxelles (Advanced ARC project ``Holography, Gauge Theories and Quantum Gravity'').

\subsection*{Data availability statement}

Data sharing is not applicable to this article as no new data were created or analyzed in this study.

\end{document}